\documentclass[3p,times]{elsarticle}

\usepackage{amssymb}
\usepackage{amsfonts,amssymb}
\usepackage{array}
\usepackage{algorithm}
\usepackage{amsmath}
\usepackage{algorithmic}
\usepackage{amsthm}
\newtheorem{theorem}{Theorem}
\newtheorem{definition}{Definition}
\newtheorem{corollary}{Corollary}
\newtheorem{property}{Property}

\usepackage{graphicx}
\usepackage{subfigure}
\usepackage{multirow}
\usepackage{booktabs}
\usepackage{bbding}
\usepackage{booktabs}
\usepackage{threeparttable}
\usepackage{float}
\usepackage[figuresright]{rotating}

\begin{document}

\begin{frontmatter}



%% e.g. \dochead{17th International Conference on Dynamical Processes in Excited States of Solids}

\title{Generation Matrix: An Embeddable Matrix Representation for Hierarchical Trees}
\tnotetext[label3]{This work was supported in part by the National	Natural Science Foundation of China (project numbers 62072109 and U1804263).}
\cortext[label4]{Corresponding author.}

\author[label1]{Jianping~Cai}
\ead{jpingcai@163.com}
\author[label1]{Ximeng~Liu\corref{label4}}
\ead{nbnix@qq.com}
\author[label1]{Jiayin~Li}
\ead{lijiayin2019@gmail.com}
\author[label2]{Shuangyue~Zhang}
\ead{zhangsy@hxxy.edu.cn}

\address[label1]{College of Computer and Data Science, Fuzhou University, Fuzhou 350108, China}
\address[label2]{College of information and Smart Electromechanical Engineering, Xiamen Huaxia University, Xiamen 361024, China}

%Computer and Data Science, Fuzhou University, Fuzhou 350108, China.
%\protect\\
%% note need leading \protect in front of \\ to get a newline within \thanks as
%% \\ is fragile and will error, could use \hfil\break instead.
%E-mail: jpingcai@163.com, snbnix@gmail.com, lijiayin2019@gmail.com.\protect\\
%S. Zhang is with the College of information and Smart Electromechanical, Xiamen Huaxia University, Xiamen 361024, China.\protect\\
%E-mail: zhangsy@hxxy.edu.cn.}% <-this % stops an unwanted space
%\thanks{Manuscript received July 13, 2021; revised XXXX XX, XXXX.}}

\begin{abstract}
Starting from the local structures to study hierarchical trees is a common research method. However, the cumbersome analysis and description make the naive method challenging to adapt to the increasingly complex hierarchical tree problems. To improve the efficiency of hierarchical tree research, we propose an embeddable matrix representation for hierarchical trees, called Generation Matrix. It can transform the abstract hierarchical tree into a concrete matrix representation and then take the hierarchical tree as a whole to study, which dramatically reduces the complexity of research. Mathematical analysis shows that Generation Matrix can simulate various recursive algorithms without accessing local structures and provides a variety of interpretable matrix operations to support the research of hierarchical trees. Applying Generation Matrix to differential privacy hierarchical tree release, we propose a Generation Matrix-based optimally consistent release algorithm (GMC). It provides an exceptionally concise process description so that we can describe its core steps as a simple matrix expression rather than multiple complicated recursive processes like existing algorithms. Our experiments show that GMC takes only a few seconds to complete a release for large-scale datasets with more than $10$ million nodes. The calculation efficiency is increased by up to $100$ times compared with the state-of-the-art schemes.
\end{abstract}


\end{frontmatter}

% \linenumbers

\section{Introduction}
\label{sec:introduction}

As a fundamental data structure, hierarchical trees are widely used in different areas, including file systems \cite{htapp:hsh}, census \cite{dt:tuc}, evolution \cite{htapp:lmf}, etc. For example, in the U.S. Census Bureau's plan to apply differential privacy to protect privacy\cite{dt:tuc}, designing a novel hierarchical tree releasing algorithm is one of the important challenges\cite{htapp:ied}. The scale of the census data is so large that we can organize them into a hierarchical tree with more than 10 million nodes. Hence, the hierarchical tree release algorithms must be specially designed and highly efficient to ensure the timely release of such large-scale data. However, the design of efficient algorithms usually requires a large amount of hierarchical tree research as a theoretical basis.

Most hierarchical tree research works naturally regard the hierarchical tree as a collection of nodes and relationships. Starting from the perspective of individual nodes and local relationships to study hierarchical trees is a common research method, called Naive Research Method. Empirically, Naive Research Method often leads to overly cumbersome analysis and abstract algorithm descriptions, mainly reflected in two aspects. On the one hand, hierarchical trees contain rich relationships, such as father, son, ancestor, descendant, sibling, or cousin, but these relationships usually lack a concrete enough description. When multiple node relationships occur in an algorithm simultaneously, the intricate node relationships will make the algorithm challenging to understand.
On the other hand, Naive Research Method focuses on the local structure of hierarchical trees rather than the overall structure. The local structure is part of the overall structure, but the studies falling into local structures may prevent researchers from solving problems from a macro perspective. Worse, the additional auxiliary symbols or indexes for describing relationships and local structures pose a significant challenge to the researchers' data analysis capabilities. Imagine a researcher facing a half-page expression with various randomly labeled symbols and subscripts; how should the next step go? Therefore, Naive Research Method is too cumbersome to solve the increasingly complex hierarchical tree problems effectively.

Considering the complexity of Naive Research Method, we adopted a Matrixing Research Method for hierarchical tree problems. Its core idea is to transform the hierarchical tree into a specific matrix representation and then embed it into the research works or algorithm designs, making the initially abstract and indescribable hierarchical tree concrete and analyzable. As a more advanced research method, similar ideas widely exist in many fields such as graph theory \cite{mxp:twe,mxp:twm,mxp:sdm,mxp:tam}, group theory \cite{mxp:gti}, and  deep learning\cite{gnn:gnn,gnn:udg}. Unlike Naive Research Method, the research object of Matrixing Research Method is the hierarchical tree itself rather than the local structure. It emphasizes avoiding visiting individual nodes as much as possible but implementing operations of the hierarchical tree by the matrix representation. Therefore, the matrix representation design is critical and directly determines whether the research can proceed smoothly. The challenges of matrix representation design are as follows.

\begin{itemize}
	\itemindent 0em
	\item[1)] A non-negligible problem is the universality of recursion in hierarchical tree algorithms, while the access of individual nodes and the description of local relationships are almost inevitable in recursion. It violates the core idea of the Matrixing Research Method. Some matrix representations \cite{mxp:twe,mxp:twm,mxp:sdm,mxp:tam} have been used in the spectral theory of trees, but there is no achievement to show that the existing matrix representations can implement recursion without accessing local structures. So that whether supporting recursion is critical to the matrix representation.
	\item[2)] We hope that the matrix representation can directly serve algorithm designs, not just a theoretical analysis tool. Therefore, the matrix representation should be succinct to ensure the efficiency of algorithms. Specifically, the space overhead of each hierarchical tree node should be constant rather than the dense matrices like Distance Matrix \cite{mxp:sdm} or Ancestral Matrix \cite{mxp:tam}.
	
\end{itemize}

\subsection{Our contributions}

Considering the challenges above, we propose an embeddable matrix representation called Generation Matrix. Generation Matrix is a lower triangular matrix containing only $2n-1$ non-zero elements (i.e., the weights of nodes and edges). Applying sparse storage technologies\cite{spm:msm}, we only need to store non-zero elements, satisfying the succinctness. Compared with others \cite{mxp:twe,mxp:twm,mxp:sdm,mxp:tam}, Generation Matrix emphasizes the application in the hierarchical tree algorithms. Our analysis of properties shows that many calculations on Generation Matrix have specific mathematical meanings. We can explain them and combine them to design complex hierarchical tree algorithms. More importantly, we demonstrate that the inverse of the Generation Matrix contains the inherent logic of recursion. Therefore, we can use Generation Matrix to simulate the top-down and bottom-up recursions without accessing the local structures. Besides, we study the relationship between Generation Matrix and some existing matrix representations and find it can be easily converted to others. It implies that Generation Matrix can be combined with the theories from other matrix representations to solve hierarchical tree problems.

To demonstrate the practicability of Generation Matrix, we introduce an application on the differentially private hierarchical tree release above. Considering the consistency problem\cite{dpht:bta,dpht:pef}, we design a Generation Matrix-based optimally consistent release algorithm (GMC) for differentially private hierarchical trees. To our knowledge, GMC is the first solution to the problem by using matrix theory. It has an exceptionally concise process description so that just a simple matrix expression can summarize the core process. GMC embodies the advantages of Generation Matrix in solving problems in cross-domain. Therefore, Generation Matrix has positive significance for promoting the development of hierarchical tree-related research and applying matrix theories to solve hierarchical tree problems.

\subsection{Organization of paper}
The rest of the paper is organized as follows. Section 2 reviews the related works of existing hierarchical tree representations and differentially private hierarchical tree release. Section 3 introduces the preliminaries of hierarchical trees and the optimally consistent release of differentially private hierarchical trees. Section 4 defines Generation Matrix, then analyzes its mathematical properties and the conversion relationship with other matrix representations. In Section 5, we show the application of Generation Matrix on differentially private hierarchical tree release and design GMC. Finally, Section 6 compares GMC with the existing technology through experiments and demonstrates its efficiency.

\section{Related Works}

The research on the hierarchical tree representations mainly concentrates on data structure and graph theory fields. In the data structure field, researchers have achieved better performance in storage \cite{dst:ucr}, query \cite{dst:srol}, or structure updating \cite{dst:srod}. However, these representations are mainly for storage in computers but do not support mathematical analysis. We can not symbolize them and use them as tools for hierarchical tree researches. In the graph theory field, many works adopt matrix representations to represent trees, including Adjacency Matrix \cite{mxp:twe}, Laplacian Matrix \cite{htapp:lmf,mxp:twm}, Distance Matrix \cite{mxp:sdm}, etc. Among them, Adjacency Matrix only describes the edge information, so that it is challenging to undertake the complex model analysis. Laplacian matrix and Distance matrix are two meaningful matrix representations widely used in spectral graph theory. However, they are for undirected graphs or trees but not suitable for representing rooted hierarchical trees.
To summarize, none of the three matrix representations is the best choice for representing hierarchical trees. Subsequently, Eric et al. \cite{mxp:tam} proposed a new matrix representation for rooted trees (i.e., hierarchical trees), named Ancestor Matrix. It represents the structure of a hierarchical tree by describing the number of overlapping edges on the path from any two leaves to the root. Studying Ancestral Matrix, Eric et al. \cite{mxp:tam} obtained many essential conclusions, such as the maximum spectrum radius and the determinants of the characteristic polynomial. However, it is also not the best choice for the calculations of hierarchical trees. First, the dense Ancestral Matrix is not succinct enough. Secondly, Ancestral Matrix is a kind of matrix with a high degree of feature summary. Although it can deterministically express the structure of a hierarchical tree only by describing the leaves, it is very unintuitive and cannot simulate the operations of hierarchical trees. Therefore, the existing matrix representations are not suitable for the analysis and calculation of hierarchical tree models. Significantly, the broad application of the Laplacian matrix in deep learning\cite{gnn:gnn,gnn:udg} in recent years implies that matrix representation has essential value for solving complex scientific problems. It motivates us to design a new matrix representation to solve hierarchical tree problems and design algorithms.

Differentially private hierarchical tree release is a data releasing technology that organizes the data into a hierarchical tree and applies differential privacy (DP) \cite{dp:cnt} to protect individual privacy. It is widely used in many scenarios, such as histogram publishing \cite{dpht:bta,dpht:uhm}, location privacy release based on spatial partitioning \cite{dpht:dps}, trajectory data publishing \cite{dpht:dpt}, frequent term discovery \cite{dpht:pef}. By adding random noise to the data, DP provides a provable and quantifiable guarantee of individual privacy. However, the random noise will destroy the consistency that the hierarchical tree should satisfy, i.e., ``the sum of the children's values equals the value at the parent''\cite{dpht:bta}. Therefore, ensuring that the released results meet consistency and obtain a higher accuracy is one of the leading research goals. Hay et al. \cite{dpht:bta} first applied a hierarchical tree to improve the accuracy of range query and designed Boosting for the consistency of histogram release. However, Boosting can only support complete $k-$ary trees, which significantly limits its application. 
Moreover, Hay et al.'s error analysis \cite{dpht:bta} of the released results is rough, and only qualitative error results are obtained. Subsequently, Wahlbeh et al. analyzed the error of Boosting and designed an algorithm to calculate the error. However, it also can only support complete $k-$ary trees. In the differentially private frequent term discovery problem studied by Ning et al. \cite{dpht:pef}, the hierarchical tree is arbitrary. Therefore, it is impossible to apply Boosting.
For this reason, Ning et al. designed an optimally consistent release algorithm for arbitrary hierarchical trees in their proposed algorithms PrivTrie \cite{dpht:pef}. Its implementation is based on multiple complex recursions, which is not easy to understand and a large number of function calls result in significant additional computational overhead. Applying the idea of maximum likelihood estimation, Lee et al. \cite{dpc:mlp} proposed a general solution for differentially private optimally consistent release. It solves the optimally consistent release by establishing a quadratic programming equation and has a closed-form matrix expression. Theoretically, it can apply to arbitrary optimally consistent release, but the computational overhead is so significant that it can only be processed for small-scale releases. However, Lee et al.'s research work \cite{dpc:mlp} is inspiring. It motivates us to try to analyze the issues of differentially private hierarchical tree release from matrix analysis.

Under the representation of Generation Matrix, we can introduce many matrix analysis methods to solve hierarchical tree problems. One of them is QR decomposition \cite{mx:laa}. In QR decomposition, we can transform any matrix into a triangular matrix by orthogonal transformation. Compared with the original form, the triangular matrix is simpler, has many exciting properties \cite{mx:tar}. The orthogonal transformation methods commonly used for QR decomposition include Householder transformation \cite{mx:rea}, Gram-Schmidt orthogonalization \cite{mx:fio}, Givens rotation \cite{mx:agr}, etc. Among them, Householder transformation is the simplest and more suitable for sparse matrices.

\section{Preliminaries}
In this section, we describe some preliminaries of our work. Before formally describing, we first introduce some of the main notation definitions shown in Tab \ref{tab:note}.

\begin{table}[!h]
	\caption{Notations Descriptions in Our Work}
	\fontsize{9}{12}\selectfont
	\centering
	\footnotesize
	\label{tab:note}
	\begin{tabular}{p{2.5cm}<{\centering}p{10.8cm}<{\raggedright}}
		\toprule[1.5pt]
		\multirow{1}{2.5cm}{\centering \textbf{Notations}}&\multirow{1}{10.8cm}{\centering \textbf{Descriptions}}\cr\midrule
		$\mathcal{T}$&Hierarchical tree with arbitrary structure\cr\hline
		${{\mathcal{T}}^{\left( k \right)}}$&$k-$order subtree of $\mathcal{T}$\cr\hline
		${{f}_{i}}$, ${{\mathcal{C}}_{i}}$&Parent of node $i$; the set of children of node $i$\cr\hline
		\multirow{1}*{$n$, ${{n}_{k}}$}&The number of nodes of the hierarchical tree; the $k-$order subtree\cr\hline
		\multirow{1}*{$h$, ${{h}_{i}}$}&The height of the hierarchy tree; The height of node $i$\cr\hline
		$m$&The number of unit counts\cr\hline
		\multirow{2}{2.5cm}{\centering $\mathbf{G}_{\mathcal{T}}^{\left( {{\mathbf{w}}_{node}},{{\mathbf{w}}_{edge}} \right)}\in {{\mathbb{R}}^{n\times m}}$}&\multirow{2}{10.8cm}{The Generation Matrix defined by a $\mathcal{T}$ with the node weights ${{\mathbf{w}}_{node}}$ and the edge weights ${{\mathbf{w}}_{edge}}$} \cr\\\hline
		\multirow{1}{2.5cm}{\centering ${{\mathbf{G}}_{\mathcal{T}}}\in {{\mathbb{R}}^{n\times n}}$}&The structure matrix of $\mathcal{T}$\cr\hline
		\multirow{1}{2.5cm}{\centering ${{\mathbf{M}}_{\mathcal{T}}}\in {{\mathbb{R}}^{n\times m}}$}&The consistency constraint matrix of $\mathcal{T}$\cr\hline
		\multirow{1}{2.5cm}{\centering ${{\mathbf{G}}_{{{\mathcal{T}}^{\left( 1 \right)}}\leftarrow \mathcal{T}}}\in {{\mathbb{R}}^{{n}_{1}\times {n}_{1}}}$}&The Generation Matrix inner-product equivalent to ${{\mathbf{M}}_{\mathcal{T}}}$\cr\hline
		${{\mathbf{A}}_{\mathcal{T}}},{{\mathbf{L}}_{\mathcal{T}}},{{\mathbf{D}}_{\mathcal{T}}}\in {{\mathbb{R}}^{n\times n}}$, ${{\mathbf{C}}_{\mathcal{T}}}\in {{\mathbb{R}}^{m\times m}}$ &\multirow{2}{10.8cm}{The Adjacency Matrix, Laplacian Matrix, Distance Matrix and Ancestral Matrix of $\mathcal{T}$}\cr\hline
		$\mathbf{x}\in {{\mathbb{R}}^{m\times 1}}$&The vector composed of unit counts ${{x}_{i}}$\cr\hline
		\multirow{1}*{$\mathbf{v}\in {{\mathbb{R}}^{n\times 1}}$}&The vector composed of the values of nodes of the hierarchical tree arranged in order\cr\hline
		\multirow{1}*{$\widetilde{\mathbf{v}}\in {{\mathbb{R}}^{n\times 1}}$, $\widetilde{\mathbf{x}}\in {{\mathbb{R}}^{m\times 1}}$}&The noisy $\mathbf{v}$ and $\mathbf{x}$ satisfying DP\cr\hline
		\multirow{1}*{$\overline{\mathbf{v}}\in {{\mathbb{R}}^{n\times 1}}$, $\overline{\mathbf{x}}\in {{\mathbb{R}}^{m\times 1}}$}&Optimally consistent release after post-processing and the vector restored from $\overline{\mathbf{v}}$\cr\hline
		\multirow{1}*{${{\mathbf{H}}_{\hbar }}\in {{\left\{ 0,1 \right\}}^{m\times n}}$}&The mapping matrix representing the mapping relationship between ${{x}_{i}}$ and ${{v}_{i}}$\cr
		\bottomrule[1.5pt]
	\end{tabular}
\end{table}

\subsection{Hierarchical Tree}

\begin{figure}[h]
	\centering
	\includegraphics[width=3.6in,trim=0 80 0 0,clip]{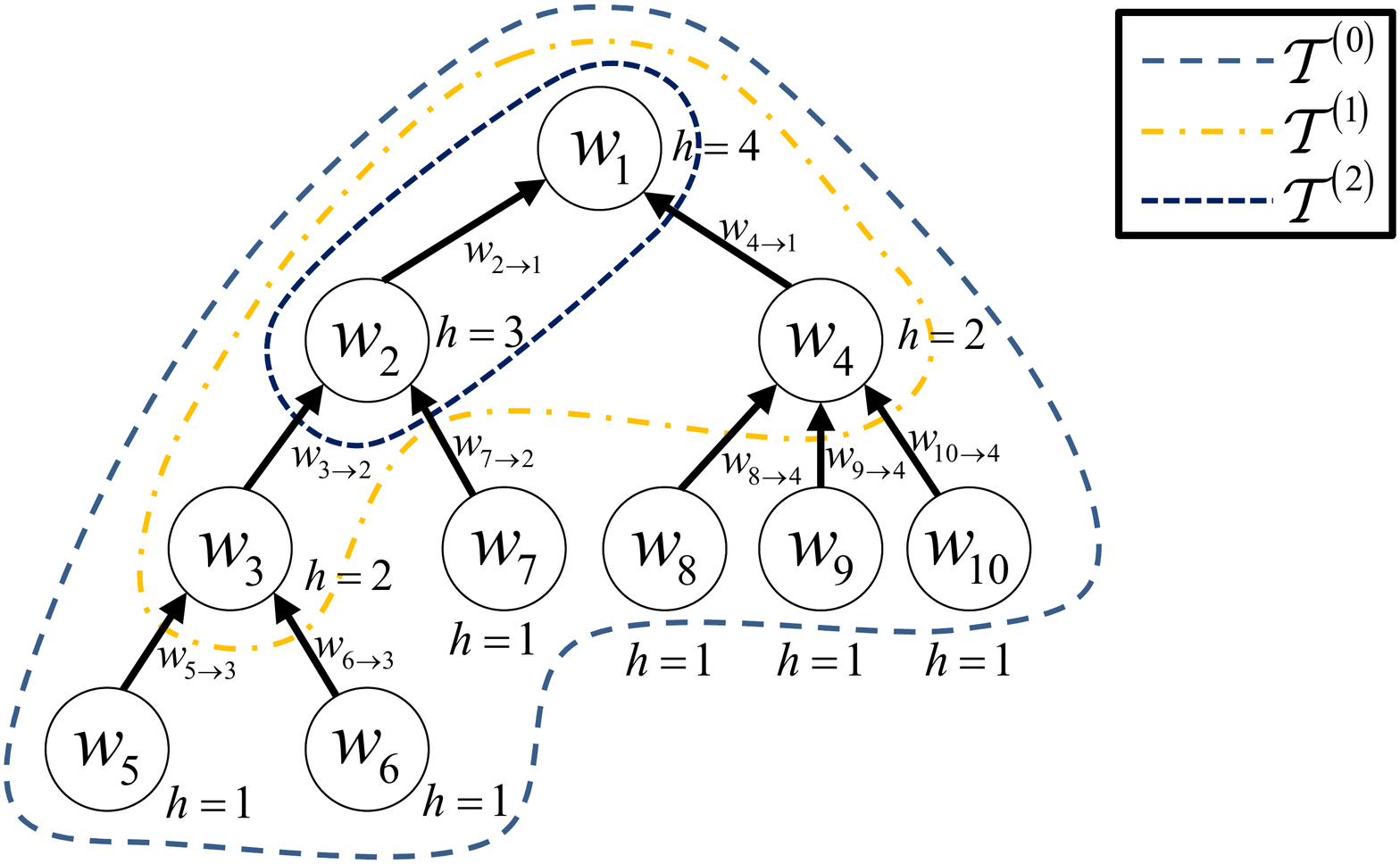}
	% where an .eps filename suffix will be assumed under latex,
	% and a .pdf suffix will be assumed for pdflatex; or what has been declared
	% via \DeclareGraphicsExtensions.
	\caption{Hierarchical Tree and Its $k-$Order Subtree}
	\label{fig:ht}
\end{figure}

We first recall the definition of the hierarchical tree.

\begin{definition}[Hierarchical Tree\cite{ht:ecv}]	
	The hierarchical tree $\mathcal{T}$ is a collection of nodes numbered $1,2,\ldots ,n$, which satisfies	
	\begin{itemize}
		\itemindent 0em
		\item[1)] $\mathcal{T}$ contains a specially designated node called the root.
		\item[2)] The remaining nodes are divided into several non-empty collections called the subtrees of the root.		
	\end{itemize}
	\label{def:ht}
\end{definition}

In the hierarchical tree, we denote the relationships between the nodes as $i\to j$, indicating node $j$ is the parent of node $i$. Fig. \ref{fig:ht} shows a weighted hierarchical tree with $10$ nodes. The node weights and edge weights are denoted as ${{w}_{i}}$ and ${{w}_{i\to j}}$, respectively. By the height of each node, we can obtain a good quasi-ranking, which is defined as follows.

\begin{definition}[Node Height]	
	The node height of a hierarchical tree refers to the height of the subtree rooted at the current node. Let the height of node $i$ be denoted as ${{h}_{i}}$, then ${{h}_{i}}$ can be calculated by the following recursive expression.
	\begin{equation}
	{{h}_{i}}=\left\{ \begin{array}{*{35}{l}}
	1 & ,\text{node }i\text{ is leaf}  \\
	\underset{j\in {{\mathcal{C}}_{i}}}{\mathop{\max }}\,{{h}_{j}}+1 & ,\text{otherwise}  \\
	\end{array} \right..
	\label{dpIeq}
	\end{equation}
	\label{def:ndH}
\end{definition}

By the height of the nodes, we define a kind of induced subtree of hierarchical trees, called $k-$Order Subtree.

\begin{definition}[$k-$Order Subtree]
	For a hierarchical tree $\mathcal{T}$ under the descending order of height, the $k-$Order Subtree ${{\mathcal{T}}^{\left( k \right)}}$ is defined as an induced subtree retained after $\mathcal{T}$ deletes all leaves $k$ times. ${{\mathcal{T}}^{\left( k \right)}}$ satisfies
	\begin{equation}
	{{\mathcal{T}}^{\left( k \right)}}=\left\{ i\left| i\in \mathcal{T}\wedge {{h}_{i}}>k \right. \right\}.
	\label{hexp}
	\end{equation}
	\label{def:kos}
\end{definition}

As a bottom-up induced subtree, ${{\mathcal{T}}^{\left( k \right)}}$ satisfies transitive, i.e., ${{\mathcal{T}}^{\left( a \right)\left( b \right)}}={{\mathcal{T}}^{\left( a+b \right)}}$. It can help us determine whether the subtrees obtained in different ways are equivalent. In subsequent applications, we use the concept of $k-$Order Subtree to simplify the description of the tree structure. In Fig. \ref{fig:ht}, we use dotted circles to mark  the subtrees of $\mathcal{T}$ from $0$ to $2$ orders. It can be seen that the $0-$Order Subtree is $\mathcal{T}$ itself actually; ${{\mathcal{T}}^{\left( 1 \right)}}$ is a subtree composed of non-leaf nodes of $\mathcal{T}$. Let ${{n}_{k}}$ denote the number of nodes contained in ${{\mathcal{T}}^{\left( k \right)}}$, then there are ${{n}_{0}}\equiv n$ and ${{n}_{1}}$ equal to the number of non-leaf nodes of $\mathcal{T}$.

\subsection{Optimally Consistent Release of Differentially Private Hierarchical Tree}

Before describing the optimally consistent releasing of the differentially private hierarchical tree, we first recall the hierarchical tree releasing model. Consider a set of unit counts ${{x}_{i}}:\mathcal{D}\to \mathbb{N}$$\left( 1\le i\le m \right)$ for private dataset $\mathcal{D}$, where ${{x}_{i}}$ indicates the number of records in $\mathcal{D}$ that satisfies the mutually exclusive unit condition ${{\varphi }_{i} }$. The unit count ${{x}_{i}}$ satisfies
\begin{equation}
{{x}_{i}}=\left| \left\{ t\in \mathbf{D}\left| {{\varphi }_{i}}\left( t \right)=True \right. \right\} \right|.
\label{unitCnt}
\end{equation}

Since ${{\varphi}_{i}}$ is mutually exclusive, any $t\in \mathcal{D}$ satisfies and only satisfies one ${{\varphi }_{i}}$. Therefore, organizing ${{x}_{i}}$ into the form of a vector, we will get $\mathbf{x}={{\left[ {{x}_{1}},{{x}_{2 }},\ldots ,{{x}_{m}} \right]}^{T}}$.

\begin{figure}[h]
	\centering
	\includegraphics[width=4.5in,trim=0 300 0 0,clip]{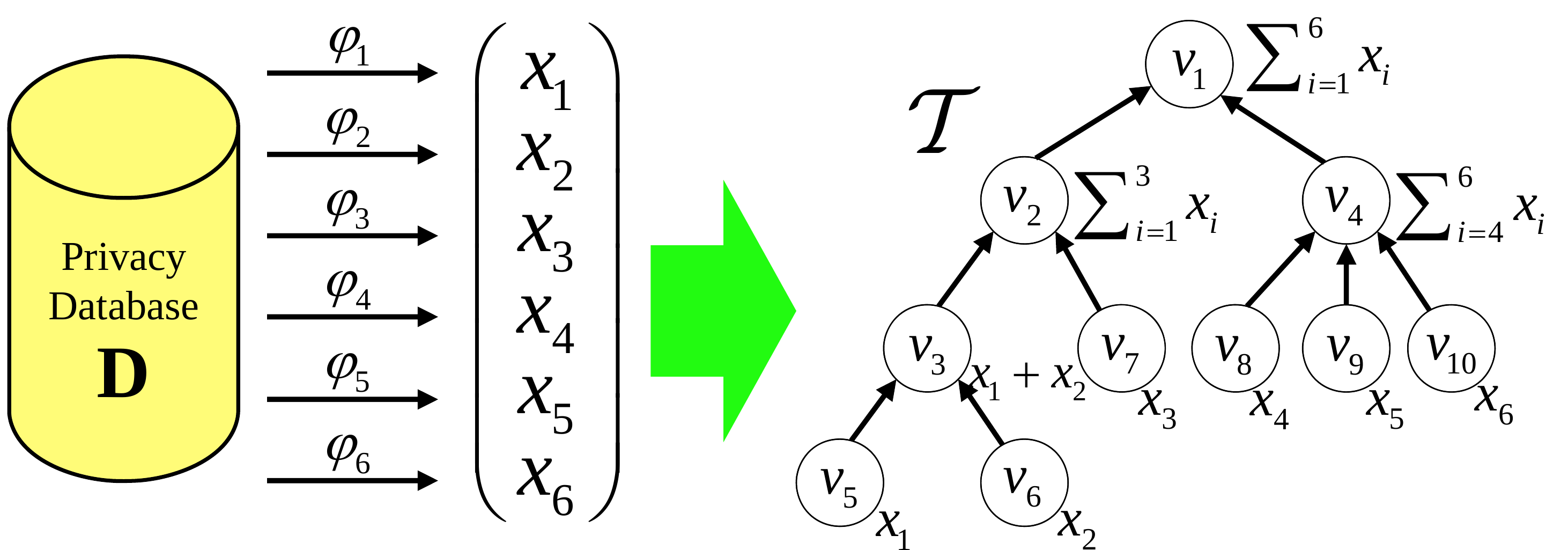}
	% where an .eps filename suffix will be assumed under latex,
	% and a .pdf suffix will be assumed for pdflatex; or what has been declared
	% via \DeclareGraphicsExtensions.
	\caption{The Process of Hierarchical Tree Release}
	\label{fig:htr}
\end{figure}
As shown in Fig. \ref{fig:htr}, each leaf corresponds to a ${{x}_{i}}$. The non-leaf node's value equals the sum of the leaves' value of the subtree rooted at that node. Therefore, the results of the hierarchical tree meet the consistency, i.e., ``the sum of the values of the child nodes is equal to the value of the parent node''. In Fig. \ref{fig:htr}, we denote the value of the $i$-th node as ${{v}_{i}}$. Then, organizing ${{v}_{i}}$ into the form of $\mathbf{v}={{\left[ {{v}_{1}},{{v}_{2}},\ldots ,{{v}_{n}} \right]}^{\operatorname{T}}}$ in turn, we can get the to-be-released data $\mathbf{v}$.

However, several works\cite{dpht:bta,dpht:uhm,dpht:dps,dpht:dpt,dpht:pef} have demonstrated that releasing an unprotected hierarchical tree will result in privacy disclosure. To protect individual privacy, DWork et al. \cite{dp:cnt} proposed differential privacy, defined as follows. 
\begin{definition}[$\varepsilon -$Differential Privacy\cite{dp:cnt}]	
	If a random algorithm $\mathcal{M}$ satisfies $\varepsilon -$difference privacy, then for any two neighboring datasets $\mathcal{D}$ and $\mathcal{D}'$, all outputs $O\in \operatorname{Range}\left( \mathcal{M} \right)$ satisfies	
	\begin{equation}
	\Pr \left( \mathcal{M}\left( \mathcal{D} \right)=O \right)\le {{e}^{\varepsilon }}\Pr \left( \mathcal{M}\left( \mathcal{D}' \right)=O \right).
	\label{dpIeq}
	\end{equation}
	\label{def:dp}
\end{definition}

Under differential privacy, the process of hierarchical tree releasing can be described as
\begin{equation}
\widetilde{\mathbf{v}}=\mathbf{v}+\mathbf{\xi },
\label{dp:addNoi}
\end{equation}
where $\widetilde{\mathbf{v}}$ is the $\mathbf{v}$ after noise addition, which satisfies differential privacy. $\mathbf{\xi }$ is the random vector for the noise addition. Each element ${{\xi }_{i}}$ is i.i.d and satisfies ${{\xi }_{i}}\sim \operatorname{Lap}\left( {\Delta }/{\varepsilon }\; \right)$, where $\operatorname{Lap}$ represents a Laplacian distribution and $\Delta $ is data sensitivity. In hierarchical tree releasing, $\Delta $ equals the height of $\mathcal{T}$\cite{dpht:bta}.

To keep the consistency of the hierarchical tree after adding noise, we can get the optimally consistent release $\overline{\mathbf{v}}$ by following the optimization equation according to maximum likelihood post-processing proposed by Lee et al. \cite{dpc:mlp}.

\begin{equation}
\begin{array}{*{35}{l}}
\underset{\overline{\mathbf{v}}}{\mathop{\min }}\, & \left\| \overline{\mathbf{v}}-\widetilde{\mathbf{v}} \right\|  \\
s.t. & {{\mathbf{M}}_{\mathcal{T}}}^{T}\overline{\mathbf{v}}=\mathbf{0}  \\
\end{array},
\label{opt:ic}
\end{equation}
where ${{\mathbf{M}}_{\mathcal{T}}}$ is the consistency constraint matrix of a hierarchical tree, defined as follows.

\begin{definition}[Consistency Constraint Matrix of Hierarchical Tree]
	
	Given a hierarchical tree $\mathcal{T}$ containing $n$ nodes. Let ${{n}_{1}}$ denote the number of non-leaf nodes in $\mathcal{T}$. The value ${{m}_{ij}}$ in row $i$ and column $j$ of the consistency constraint matrix ${{\mathbf{M}}_{\mathcal{T}}}\in {{\mathbb{R}}^{n\times {{n}_{1}}}}$ is defined as follows:
	\begin{equation}
	{{m}_{ij}}=\left\{ \begin{array}{*{35}{l}}
	1 & ,i=j  \\
	-1 & ,j={{f}_{i}}  \\
	0 & ,\text{otherwise}  \\
	\end{array} \right.,	 
	\label{vCCM}
	\end{equation}
	where ${{f}_{i}}$ is the parent of node $i$.
	\label{def:cmm}
\end{definition}

The optimization equation \eqref{opt:ic} has the following closed-form expression:
\begin{equation}
\overline{\mathbf{v}}=\widetilde{\mathbf{v}}-{{\mathbf{M}}_{\mathcal{T}}}{{\left( {{\mathbf{M}}_{\mathcal{T}}}^{T}{{\mathbf{M}}_{\mathcal{T}}} \right)}^{-1}}{{\mathbf{M}}_{\mathcal{T}}}^{T}\widetilde{\mathbf{v}} .
\label{exp:ic}
\end{equation}

Since Formula \eqref{exp:ic} involves the inner product and inverse operations of the matrix, the time complexity of the direct solution is as high as $O\left( {{n}^{3}} \right)$. The amount of calculation is too large to obtain an efficient enough algorithm directly by the expressions. On the surface, Formula \eqref{exp:ic} is not a good choice for solving optimally consistent releases, but under the theories of Generation Matrix, we can convert it into another form and apply the properties of Generation Matrix to obtain an efficient algorithm.

\section{Generation Matrix Model for Hierarchical Tree}
\subsection{Generation Matrix}
Before defining Generation Matrix, we number the nodes of $\mathcal{T}$ by descending order of height firstly.

\begin{definition}[Descending Order of Height]
	Let ${{h}_{i}}$ denote the height of node $i$ defined by Def. \ref{def:ndH}. If any two nodes $i$ and $j$ in $\mathcal{T}$ satisfy 
	\begin{equation}
	i<j\Rightarrow {{l}_{i}}\ge {{l}_{j}},
	\label{hcond}
	\end{equation}
	we say that $\mathcal{T}$ satisfies Descending Order of Height.
	\label{def:doh}
\end{definition}
Under the descending order of height, we define Generation Matrix for $\mathcal{T}$.

\begin{figure}[h]
	\centering
	\includegraphics[width=5in,trim=0 270 0 0,clip]{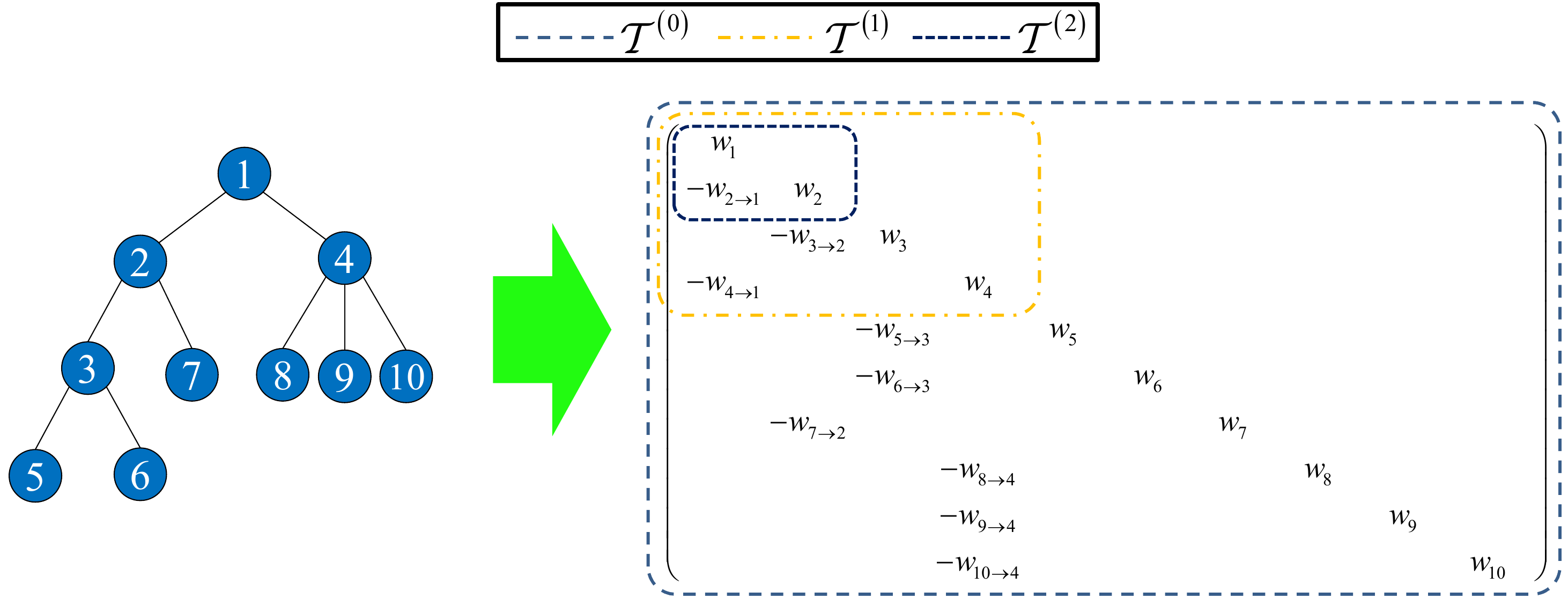}
	% where an .eps filename suffix will be assumed under latex,
	% and a .pdf suffix will be assumed for pdflatex; or what has been declared
	% via \DeclareGraphicsExtensions.
	\caption{Generation Matrix and Its $k-$Order Submatrix}
	\label{fig:gm}
\end{figure}

\begin{definition}[Generation Matrix]
	Considering a non-zero weighted hierarchical tree $\mathcal{T}$ under descending order of height, let ${{w}_{i}}$ and ${{w}_{i\to {{f}_{i}}}}$$\left( {{w}_{i}},{{w}_{i\to {{f}_{i}}}}\ne 0 \right)$ denote the weight values of the node $i$ and the edge $i\to {{f}_{i }}$. Organizing them into a vector, denoted as ${{\mathbf{w}}_{node}}$ and ${{\mathbf{w}}_{edge}}$, the generation matrix is denoted as $\mathbf{G}_{\mathcal{T}}^{\left( {{\mathbf{w}}_{node}},{{\mathbf{w}}_{edge}} \right)}\in {{\mathbb{R}}^{n\times n}}$, whose element ${{g}_{i,j}}$ in row $i$ and column $j$ is defined as follows,
	\begin{equation}
	{{g}_{i,j}}=\left\{ \begin{array}{*{35}{l}}
	{{w}_{i}} & ,i=j  \\
	-{{w}_{i\to {{f}_{i}}}} & ,i\to {{f}_{i}}  \\
	0 & ,\text{otherwise}  \\
	\end{array} \right..
	\label{hexp}
	\end{equation}
	\label{def:gm}
\end{definition}

As shown in Fig. \ref{fig:gm}, since $\mathcal{T}$ satisfies Descending Order of Height, the number of any non-root node $i$ in $\mathcal{T}$ is always bigger than its parent. It ensures Generation Matrix is always a lower triangular matrix. According to Def. \ref{def:gm}, there is a one-to-one mapping between arbitrary non-zero weighted hierarchical tree and $\mathbf{G}_{\mathcal{T}}^{\left( {{\mathbf{w}}_{node}},{{\mathbf{w}}_{edge}} \right)}\in {{\mathbb{R}}^{n\times n}}$, i.e., the matrix representation of the non-zero weighted hierarchical tree is unique. When we only need to describe the structure of the hierarchical tree, we can use a Generation Matrix with all weights of $1$ to represent it, i.e., $\mathbf{G}_{\mathcal{T}}^{\left( \mathbf{1},\mathbf{1} \right)}$. We call $\mathbf{G}_{\mathcal{T}}^{\left( \mathbf{1},\mathbf{1} \right)}$ \textbf{Structure Matrix}, which is abbreviated as ${{\mathbf{G}}_{\mathcal{T}}}$. If two hierarchical trees have the same structure and arrangement, the positions of non-zero elements in their Generation Matrices are always the same, which we call Similar Generation Matrices.

\begin{definition}[Similar Generation Matrices]
	
	If ${{\mathbf{G}}_{1}}$ and ${{\mathbf{G}}_{2}}$ are two Generation Matrices defined by the hierarchical trees with the same structure and arrangement (or the same tree), we call them similar Generation Matrices, which are denoted as ${{\mathbf{G}}_{1}}\sim {{\mathbf{G}}_{2}}$.
	\label{def:sgm}
\end{definition}

Every Generation Matrix from the same hierarchical tree is always similar. We can use $\mathbf{G}\sim {{\mathbf{G}}_{\mathcal{T}}}$ as sufficient to judge whether the hierarchical tree represented by $\mathbf{G}$ has the same structure as $\mathcal{T}$.

According to Def. \ref{def:kos}, ${{\mathcal{T}}^{\left( k \right)}}$ is an induced subtree composed of nodes $i$ with ${{h}_{i}}\ge k$ in $\mathcal{T}$. Under Descending Order of Height, these nodes are always ranked first. 
In Fig. \ref{fig:gm}, $k-$Order Submatrix that represents the $k-$Order Subtree is denoted as $\mathbf{G}_{{{\mathcal{T}}^{\left( k \right)}}}^{\left( {{\mathbf{w}}_{node}},{{\mathbf{w}}_{edge}} \right)}$. We can obtain it by taking the ${{n}_{k}}-$order leading principal minor of $\mathbf{G}_{\mathcal{T}}^{\left( {{\mathbf{w}}_{node}},{{\mathbf{w}}_{edge}} \right)}$ (i.e., the elements in rows and columns from $1$ to ${{n}_{k}}$). In particular, $\mathbf{G}_{{{\mathcal{T}}^{\left( 1 \right)}}}^{\left( {{\mathbf{w}}_{node}},{{\mathbf{w}}_{edge}} \right)}$ represents the sub-tree composed of non-leaf nodes of $\mathcal{T}$.

Considering a specific application, one problem we may encounter is that the nodes of the hierarchical tree are numbered but not in Descending Order of Height. Under the matrix representation, the problem is elementary to solve. We can adopt a sparse mapping matrix to convert the original number into Descending Order of Height.

\begin{definition}[Mapping Matrix]
	Given an ordered set $\mathcal{S}=\left\langle {{s}_{1}},{{s}_{2}},\ldots ,{{s}_{m}} \right\rangle $ represents the mapping relationship between integers, satisfying ${{s}_{i}}\in {{\mathbb{N}}^{+}}\wedge {{s}_{i}}\le n$, the mapping matrix defined by $\mathcal{S}$ is denoted as ${{\mathbf{H}}_{\mathcal{S}}}\in {{\left\{ 0,1 \right\}}^{m\times n}}$. The value ${{h}_{i,j}}$ in row $i$ and column $j$ of ${{\mathbf{H}}_{\mathcal{S}}}$ satisfies
	\begin{equation}
	{{h}_{i,j}}=\left\{ \begin{array}{*{35}{l}}
	1 & ,{{s}_{i}}=j  \\
	0 & ,\text{otherwise}  \\
	\end{array} \right.	 .
	\label{mmEq}
	\end{equation}
	\label{def:mpm}
\end{definition}

By the definition of the mapping matrix, $\mathcal{S}$ always represents an injection and satisfies ${{\mathbf{H}}_{\mathcal{S}}}{{\mathbf{H}}_{ \mathcal{S}}}^{T}=\mathbf{I}$. If $\mathcal{S}$ represents a bijection, ${{\mathbf{H}}_{\mathcal{S}}}$ will be a permutation matrix. When the basis vector ${{\mathbf{e}}_{i}}$ acts on ${{\mathbf{H}}_{\mathcal{S}}}$, the following equations holds.
\begin{equation}
{{\mathbf{H}}_{\mathcal{S}}}^{T}{{\mathbf{e}}_{i}}={{\mathbf{e}}_{{{s}_{i}}}},
\label{mmPro1}
\end{equation}
\begin{equation}
{{\mathbf{H}}_{\mathcal{S}}}{{\mathbf{e}}_{i}}=\left\{ \begin{array}{*{35}{l}}
{{\mathbf{e}}_{s_{i}^{-1}}} & ,i\in \mathcal{S}  \\
\mathbf{0} & ,i\notin \mathcal{S}  \\
\end{array} \right.,
\label{mmPro2}
\end{equation}
where $s_{i}^{-1}$ represents the inverse mapping of ${{s}_{i}}$, which satisfies $s_{{{s}_{i}}}^{- 1}=i$.

To describe the mapping relationship between the nodes before and after sorting by Descending Order of Height, we only need to define the ordered set $\mathcal{S}=\left\langle {{s}_{1}},{{s}_{2}},\ldots ,{{s}_{n}} \right\rangle$, where ${{s}_{i}}$ represents the sorted number of the node initially numbered $i$. Then, we can use ${{\mathbf{H}}_{\mathcal{S}}}\mathbf{v}$ to get the vector before sorting from $\mathbf{v}$ after sorting.

In addition, the mapping matrix can be used to represent other mapping relationships, such as the mapping of ${{v}_{i}}$ to ${{x}_{i}}$. For example, in Fig. \ref{fig:ht}, if we use an ordered set $\mathcal{H}=\left\langle {{\hbar }_{1}},{{\hbar }_{2}},\ldots ,{{\hbar }_{m}} \right\rangle$ to represent the mapping relationship between ${{v}_{i}}$ and ${{x}_{i}}$, then ${{\hbar }_{i}}=i+5$, and we have the mapping matrix ${{\mathbf{H}}_{\hbar }}$ to represent their mapping relationship.

\subsection{Properties of Generation Matrix}

Our research shows that Generation Matrix has many mathematical properties that deserve attention. These properties can help us solve various problems about the analysis and calculation of hierarchical trees. According to Def. \ref{def:gm}, it is not difficult to find that Generation Matrix satisfies sparsity.
\begin{property}[Sparsity]
	Considering a hierarchical tree $\mathcal{T}$ consists of $n$ nodes, $\mathbf{G}_{\mathcal{T}}^{\left( {{\mathbf{w}}_{node}},{{\mathbf{w}}_{edge}} \right)}$ is has and only has $2n-1$ non-zero elements. Its first row has only $1$ non-zero element, and the remaining $n-1$ rows have $2$ non-zero elements. 
	\label{prop:sp}
\end{property}

Due to the sparsity of Generation Matrix, we can apply various sparse matrix technologies such as COO (Coordinate Format) and CSR (Compressed Sparse Row) to calculate hierarchical tree models efficiently. Under the sparse representations, the storage and calculation of Generation Matrix are both only $O\left( n \right)$. Therefore, the application based on Generation Matrix does not cause more computational overhead. Currently, the computing technologies of sparse matrices are very mature and widely used in various high-performance platforms \cite{spm:its,spm:tlm,spm:cca}. 

One of the fundamental properties of Generation Matrices is invertibility. By solving the equations $\mathbf{G}{{_{\mathcal{T}}^{\left( {{\mathbf{w}}_{node}},{{\mathbf{w}}_{edge}} \right)}}^{T}}\mathbf{z}=\mathbf{v}$ and $\mathbf{G}_{\mathcal{T}}^{\left( {{\mathbf{w}}_{node}},{{\mathbf{w}}_{edge}} \right)}\mathbf{z}=\mathbf{v}$ about $\mathbf{z}$, we find two interesting and important mathematical properties of Generation Matrices. We collectively call them the propagation of Generation Matrix.

\begin{property}[Upward Propagation]
	Let ${{g}_{i,j}}$ denotes the element in row $i$ and column $j$ of $\mathbf{G}_{\mathcal{T}}^{\left( {{\mathbf{w}}_{node}},{{\mathbf{w}}_{edge}} \right)}$, and ${{v}_{i}}$ is the value of node $i$ of $\mathcal{T}$. Organize ${{v}_{i}}$ into vector $\mathbf{v}={{\left[ {{v}_{1}},{{v}_{2} },\ldots ,{{v}_{n}} \right]}^{T}}$, then $\mathbf{z}=\mathbf{G}{{_{\mathcal{T}}^{\left( {{\mathbf{w}}_{node}},{{\mathbf{w}}_{edge}} \right)}}^{-T}}\mathbf{v}$ is an upward propagation on $\mathbf{v}$. The value ${{z}_{i}}$ of $\mathbf{z}$ satisfies
	\begin{equation}
	{{z}_{i}}=\left\{ \begin{array}{*{35}{l}}
	{{{v}_{i}}}/{{{g}_{i,i}}}\; & ,\text{node }i\text{ is leaf}  \\
	{\left( {{v}_{i}}-\sum\limits_{j\in {{\mathcal{C}}_{i}}}{{{g}_{j,i}}{{z}_{j}}} \right)}/{{{g}_{i,i}}}\; & ,\text{otherwise}  \\
	\end{array} \right.	.
	\label{hexp}
	\end{equation}
	\label{prop:up}
\end{property}
\begin{property}[Downward Propagation]
	If $\mathbf{G}_{\mathcal{T}}^{\left( {{\mathbf{w}}_{node}},{{\mathbf{w}}_{edge}} \right)}$ and $\mathbf{v}$ have the exact definition as Prop. \ref{prop:up}, then $\mathbf{z}=\mathbf{G}{{_{\mathcal{T}}^{\left( {{\mathbf{w}}_{node}},{{\mathbf{w}}_{edge}} \right)}}}^{-1}\mathbf{v}$ is a downward propagation on $\mathbf{v}$. The value ${{z}_{i}}$ of $\mathbf{z}$ satisfies
	\begin{equation}
	{{z}_{i}}=\left\{ \begin{array}{*{35}{l}}
	{{{v}_{i}}}/{{{g}_{i,i}}}\; & ,\text{node }i\text{ is root}  \\
	{\left( {{v}_{i}}-{{g}_{i,{{f}_{i}}}}{{z}_{{{f}_{i}}}} \right)}/{{{g}_{i,i}}}\; & ,\text{otherwise}  \\
	\end{array} \right..
	\label{hexp}
	\end{equation}
	\label{prop:dp}
\end{property}

Prop. \ref{prop:up} and Prop. \ref{prop:dp} show that Generation Matrix can simulate multiple recursive operations of hierarchical trees. According them, Cor. \ref{cor:infEle} analyzes the elements affected by the propagations of Generation Matrix.

\begin{corollary}
	For the upward propagation, affected ${{z}_{j}}$ by ${{v}_{i}}$ satisfies $j=i$, or $j$ is the ancestor of $i$ in $\mathcal{T}$; for the downward propagation, affected ${{z}_{j}}$ by ${{v}_{i}}$ satisfies $j=i$, or $j$ is a descendant of $i$ in $\mathcal{T}$.
	\label{cor:infEle}
\end{corollary}

Combining the propagations, we further study a variety of matrix operations of Generation Matrices. They have strong interpretability and provide crucial theoretical support for the research of hierarchical trees.

\begin{property}
	Let the vector $\mathbf{z}=\left( \mathbf{I}-{{\mathbf{G}}_{\mathcal{T}}}^{T} \right)\mathbf{1}$, then the $i$-th element ${{z}_{i}}$ of $\mathbf{z}$ represents the number of children of the node $i$, i.e., ${{z}_{i}}=\left| {{\mathcal{C}}_{i}} \right|$.
	\label{prop:chn}
\end{property}

\begin{property}
	If the vector $\mathbf{z}={{\mathbf{G}}_{\mathcal{T}}}^{-T}\mathbf{1}$,  the $i$-th element ${{z}_{i}}$ of $\mathbf{z}$ represents the number of nodes contained in the subtree rooted at node $i$.
	\label{prop:stn}
\end{property}

\begin{property}
	Let the vector $\mathbf{z}={{\mathbf{G}}_{\mathcal{T}}}^{-1}\mathbf{1}$, then the $i$-th element ${{z}_{i}}$  of $\mathbf{z}$ represents the depth of node $i$, where the depth of the root is $1$.
	\label{prop:depn}
\end{property}

The properties above indicate that Generation Matrix is an effective and easy-to-use tool for various hierarchical tree analyses. In addition to the conclusions about vectors discussed above, there are some conclusions about matrices as follows. Compared with conclusions about vectors, they focus on describing the characteristics between nodes.

\begin{property}
	Let $g_{ij}^{\left( -1 \right)}$ denote the element in row $i$ and column $j$ of ${{\mathbf{G}}_{\mathcal{T}}}^{-1}$,  then $g_{ij}^{\left( -1 \right)}$ satisfies
	\begin{equation}
	g_{ij}^{\left( -1 \right)}=\left\{ \begin{array}{*{35}{l}}
	1 & ,i=j\vee j\text{ is an ancestor of }i  \\
	0 & ,\text{otherwise}  \\
	\end{array} \right..
	\label{hexp}
	\end{equation}
	\label{prop:inv}
\end{property}
\begin{proof}
	%	See Appendix \ref{prov:prop:inv}.
	See Appendix  A.1.
\end{proof}

Prop. \ref{prop:inv} shows that ${{\mathbf{G}}_{\mathcal{T}}}^{-1}$ is a matrix indicating the relationship between ancestors and descendants. Although ${{\mathbf{G}}_{\mathcal{T}}}^{-1}$ is denser than ${{\mathbf{G}}_{\mathcal{T}}}$, in most cases, ${{\mathbf{G}}_{\mathcal{T}}}^{-1}$ is still sparse.

\begin{property}
	Let $\mathbf{M}={{\mathbf{G}}_{\mathcal{T}}}{{\mathbf{G}}_{\mathcal{T}}}^{T}$ and ${{m}_{i,j}}$ denote the element in row $i$ and column $j$ of $\mathbf{M}$. Except for ${{m}_{11}}=1$, other elements satisfy ``${{m}_{ij}}=1$ $\Leftrightarrow$ $i$ and $j$ are sibling nodes''.
	\label{prop:sib}
\end{property}

\begin{property}
	Let $\mathbf{M}={{\left( {{\mathbf{G}}_{\mathcal{T}}}^{T}{{\mathbf{G}}_{\mathcal{T}}} \right)}^{-1}}$ and ${{m}_{ij}}$ is the element in row $i$ and column $j$ of $\mathbf{M}$, then the value of ${{m}_{ij}}$ represents the number of common ancestors of the node pair $i$ and $j$, and ${{m}_{ii}}$ represents the depth of $i$.
	\label{prop:can}
\end{property}
\begin{proof}
	%	See Appendix \ref{prov:prop:can}.
	See Appendix A.2.
\end{proof}

Similar to Prop. \ref{prop:inv}, Prop. \ref{prop:sib} can also be used as an indicator matrix to describe the relationship between nodes. Prop. \ref{prop:can} is an important property, which describes an effective method of calculating common ancestors. As an essential feature to describe the correlation between nodes, the number of common ancestors has an important application value for hierarchical tree analyses\cite{mxp:tam}.

In the study of spectral graph theory, feature analysis is usually indispensable. Our research shows that the eigenvalues and eigenvectors of a Generation Matrix satisfy the following properties.

\begin{property}[Eigenvalues and Eigenvectors]
	Let $\boldsymbol{\lambda }={{\left[ {{\lambda }_{1}},{{\lambda }_{2}},\ldots ,{{\lambda }_{n}} \right]}^{T}}$ denote the eigenvalues of $\mathbf{G}_{\mathcal{T}}^{\left( {{\mathbf{w}}_{node}},{{\mathbf{w}}_{edge}} \right)}$, then the $i$-th eigenvalue${{\lambda }_{i}}$ is $w_{i}$.
	
	Let the left eigenvector and the right eigenvector corresponding to the $i$-th eigenvalue denote as ${{\mathbf{u}}_{i}}$ and ${{\mathbf{v}}_{i}}$, respectively. The premise of the existence of ${{\mathbf{u}}_{i}}$ is that each ancestor $j$ of $i$ satisfies $ w_{i}\ne w_{j}$, and the premise of the existence of ${{\mathbf{v}}_{i}}$ is that each descendant $j$ of $i$ satisfies $ w_{i}\ne w_{j}$.
	
	Let the $j$-th element of ${{\mathbf{u}}_{i}}$ and ${{\mathbf{v}}_{i}}$ denote as $u_{j}^{\left( i \right)}$ and $v_{j}^{\left( i \right)}$, respectively. If ${{\mathbf{u}}_{i}}$ exists, then $u_{j}^{\left( i \right)}=0$ for any $j>i$. Let $u_{i}^{\left( i \right)}=1$. The remaining elements $u_{j}^{\left( i \right)}\left( j<i \right)$ can be obtained by 
	\begin{equation}
	u_{j}^{\left( i \right)}=\left\{ \begin{array}{*{35}{l}}
	\frac{\sum\nolimits_{k\in {{\mathcal{C}}_{j}}}{w_{k\to j}u_{k}^{\left( i \right)}}}{w_{j}-w_{i}} & ,j\text{ is a ancestor of }i  \\
	0 & ,\text{otherwise}  \\
	\end{array} \right..
	\label{hexp}
	\end{equation}
	If ${{\mathbf{v}}_{i}}$ exists, then $v_{j}^{\left( i \right)}=0$ for any $j<i$. Let $v_{i}^{\left( i \right)}=1$. The remaining elements $v_{j}^{\left( i \right)}\left( j>i \right)$ can be obtained by
	\begin{equation}
	v_{j}^{\left( i \right)}=\left\{ \begin{array}{*{35}{l}}
	\frac{w_{j\to {{f}_{j}}}v_{{{f}_{j}}}^{\left( k \right)}}{w_{j}-w_{i}} & ,j\text{ is a descendant of }i  \\
	0 & ,\text{otherwise}  \\
	\end{array} \right..
	\label{hexp}
	\end{equation}
	
	\label{prop:eig}
\end{property}

Prop. \ref{prop:eig} shows that the eigenvalues and eigenvectors of the Generation Matrix have many interesting properties. For example, the eigenvalue of Generation Matrix is the weights of the nodes, which is much easier to solve than other matrix representations; the eigenvectors also satisfy some propagation properties similar to Prop. \ref{prop:up} and \ref{prop:dp}. Notably, the eigenvectors are conditional, which means that Generation Matrix is not always diagonalizable. Some Generation Matrices, especially the eigenvectors of ${{\mathbf{G}}_{\mathcal{T}}}$, still have many problems waiting to be studied. Although feature analysis is not the main focus in this paper, Prop. \ref{prop:eig} still provides some valuable references for the subsequent research works.

Considering the relationship between $\mathbf{G}_{\mathcal{T}}^{\left( {{\mathbf{w}}_{node}},{{\mathbf{w}}_{edge}} \right)}$ and corresponding ${{\mathbf{G}}_{\mathcal{T}}}$, we find that $\mathbf{G}_{\mathcal{T}}^{\left( {{\mathbf{w}}_{node}},{{\mathbf{w}}_{edge}} \right)}$ satisfies a particular  decomposition form, which we call the diagonal decomposition of Generation Matrix.

\begin{property} [Diagonal Decomposition]
	Given arbitrary $\mathbf{G}_{\mathcal{T}}^{\left( {{\mathbf{w}}_{node}},{{\mathbf{w}}_{edge}} \right)}$, there is always a pair of vectors $\boldsymbol{\alpha }, \boldsymbol{\beta }\in {{\mathbb{R}}^{n}}$, making the following decomposition hold for $\mathbf{G}_{\mathcal{T}}^{\left( {{\mathbf{w}}_{node}},{{\mathbf{w}}_{edge}} \right)}$ and the structure matrix ${{\mathbf{G}}_{\mathcal{T}}}$.
	\begin{equation}
	\mathbf{G}_{\mathcal{T}}^{\left( {{\mathbf{w}}_{node}},{{\mathbf{w}}_{edge}} \right)}=\operatorname{diag}\left( \boldsymbol{\beta } \right){{\mathbf{G}}_{\mathcal{T}}}\operatorname{diag}\left( \boldsymbol{\alpha } \right)
	\label{eq:ddEq}
	\end{equation}
	Let ${{\alpha }_{i}}$ and ${{\beta }_{i}}$ denote the $i$-th element of them, respectively. Then a pair of legal $\boldsymbol{\alpha }$ and $\boldsymbol{\beta }$ can be obtained by
	\begin{equation}
	\left\{ \begin{aligned}
	& \boldsymbol{\alpha }=\exp \left( {{\mathbf{G}}_{\mathcal{T}}}^{-1}\left( \ln \left( {{\mathbf{w}}_{node}} \right)-\ln \left( {{\mathbf{w}}_{edge}} \right) \right) \right) \\ 
	& \boldsymbol{\beta }={{\mathbf{w}}_{node}}\oslash \boldsymbol{\alpha } \\ 
	\end{aligned} \right..
	\label{eq:svab}
	\end{equation}
	Where ``$\oslash$'' denote the element-wise division of vectors. Note, since the root numbered $1$ has no parent, we set ${{w}_{1\to \varnothing }}=1$ as the first element of ${{\mathbf{w}}_{edge}}$, and ${{w}_{i\to {{f}_{i}}}}$ is the $i$-th element of ${{\mathbf{w}}_{edge}}$ in the remaining elements.
	\label{prop:dd}
\end{property}
\begin{proof}
	%	See Appendix \ref{prov:prop:dd}.
	See Appendix  A.3.
\end{proof}

By the diagonal decomposition of Generation Matrix, we can express arbitrary $\mathbf{G}_{\mathcal{T}}^{\left( {{\mathbf{w}}_{node}},{{\mathbf{w}}_{edge}} \right)}$ as an expression with ${{\mathbf{G}}_{\mathcal{T}}}$. Using Prop. \ref{prop:dd}, we can extend some mathematical properties of ${{\mathbf{G}}_{\mathcal{T}}}$ to $\mathbf{G}_{\mathcal{T}}^{\left( {{\mathbf{w}}_{node}},{{\mathbf{w}}_{edge}} \right)}$ to solve more problems effectively.

\subsection{The Conversion between Generation Matrix and Other Matrix Representations}

\begin{figure}[h]
	\centering
	\includegraphics[width=5in,trim=0 185 0 0,clip]{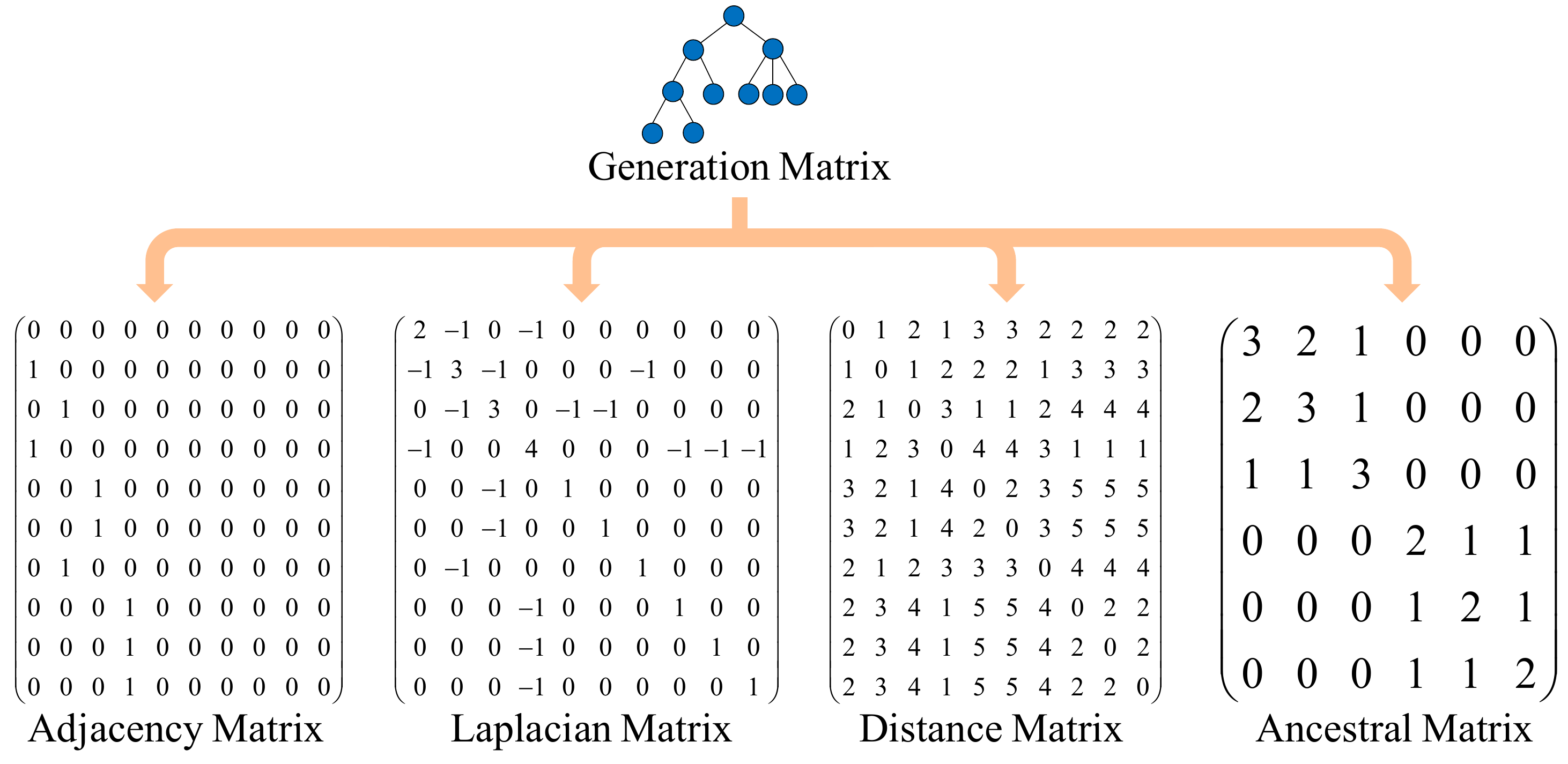}
	% where an .eps filename suffix will be assumed under latex,
	% and a .pdf suffix will be assumed for pdflatex; or what has been declared
	% via \DeclareGraphicsExtensions.
	\caption{Conversion Relationships from Generation Matrix to Other Matrix Representations}
	\label{fig:cvr}
\end{figure}

Our research shows that Generation Matrix is not an isolated matrix representation from others. Through proper operations, we can convert Generation Matrix into other matrix representations. Fig. \ref{fig:cvr} shows the four matrix representations that can be transformed by the Generation Matrix constructed by the hierarchical tree in Fig. \ref{fig:ht}, including Adjacency Matrix, Laplacian Matrix, Distance Matrix, and Ancestor Matrix.

\begin{theorem}
	Let ${{\mathbf{A}}_{\mathcal{T}}}$ be Adjacency Matrix of $\mathcal{T}$, then ${{\mathbf{A}}_{\mathcal{T}}}$ can be obtained by the following expression of ${{\mathbf{G}}_{\mathcal{T}}}$:	
	\begin{equation}
	{{\mathbf{A}}_{\mathcal{T}}}=\mathbf{I}-{{\mathbf{G}}_{\mathcal{T}}}.
	\label{hexp}
	\end{equation}
	\label{thm:cam}
\end{theorem}

\begin{theorem}
	Let ${{\mathbf{L}}_{\mathcal{T}}}$ be Laplacian Matrix of $\mathcal{T}$, then ${{\mathbf{L}}_{\mathcal{T}}}$ can be obtained by the following expression of ${{\mathbf{G}}_{\mathcal{T}}}$:	
	\begin{equation}
	{{\mathbf{L}}_{\mathcal{T}}}={{\mathbf{G}}_{\mathcal{T}}}^{T}{{\mathbf{G}}_{\mathcal{T}}}-{{\mathbf{e}}_{1}}{{\mathbf{e}}_{1}}^{T}.
	\label{exp:gclm}
	\end{equation}
	\label{thm:clm}
\end{theorem}
\begin{theorem}
	Let ${{\mathbf{D}}_{\mathcal{T}}}$ be Distance Matrix of $\mathcal{T}$, then ${{\mathbf{D}}_{\mathcal{T}}}$ can be obtained by the following expression of ${{\mathbf{G}}_{\mathcal{T}}}$:	
	\begin{equation}
	{{\mathbf{D}}_{\mathcal{T}}}={{\mathbf{G}}_{\mathcal{T}}}^{-1}\mathbf{I}{{\mathbf{I}}^{T}}+\mathbf{I}{{\mathbf{I}}^{T}}{{\mathbf{G}}_{\mathcal{T}}}^{-T}-2{{\left( {{\mathbf{G}}_{\mathcal{T}}}^{T}{{\mathbf{G}}_{\mathcal{T}}} \right)}^{-1}}.
	\label{exp:gcdm}
	\end{equation}
	\label{thm:cdm}
\end{theorem}
\begin{theorem}
	Let ${{\mathbf{C}}_{\mathcal{T}}}$ be Ancestral Matrix\cite{mxp:tam} of $\mathcal{T}$, then ${{\mathbf{C}}_{\mathcal{T}}}$ can be obtained by the following expression of ${{\mathbf{G}}_{\mathcal{T}}}$:	
	\begin{equation}
	{{\mathbf{C}}_{\mathcal{T}}}={{\mathbf{H}}_{\mathcal{H}}}{{\left( {{\mathbf{G}}_{\mathcal{T}}}^{T}{{\mathbf{G}}_{\mathcal{T}}} \right)}^{-1}}{{\mathbf{H}}_{\mathcal{H}}}^{T}-1.
	\label{exp:gccm}
	\end{equation}
	\label{thm:ccm}
\end{theorem}
\begin{proof}
	The proofs of Thm. 2-4 in Appendix A.4 to A.6.
\end{proof}

It can be seen from the theorems above that we can convert Generation Matrix into other matrix representations by simple matrix expressions. However, the reverse is not easy. Except for Adjacency Matrix, other matrix representations cannot be directly converted back to Generation Matrix. Therefore, we can use Generation Matrix to construct other matrix representations. Besides, it also implies that the theories of Generation Matrix have a particular internal connection with the matrix represented. We can combine the theories of Generation Matrix and other matrix representations to solve more problems about hierarchical trees. 

\section{The Application on Differentially Private Hierarchical Tree Release}
\subsection{Hierarchical Tree Release Based on Generation Matrix}
In this section, we introduce how to apply Generation Matrix to efficiently and concisely achieve an optimally consistent release on differentially private hierarchical tree release. Since each leaf of ${\mathcal{T}}$ corresponds to a ${{x}_{i}}$, we use the mapping matrix ${{\mathbf{H}}_{\mathcal{H}}}$ to represent the mapping relationship between leaf nodes and hierarchical tree nodes. Using Prop. \ref{prop:up}, the hierarchical tree building process ${{\operatorname{BuildTree}}_{\mathcal{T}}}$ can be described as 
\begin{equation}
\mathbf{v}={{\operatorname{BuildTree}}_{\mathcal{T}}}\left( \mathbf{x} \right)={{\mathbf{G}}_{\mathcal{T}}}^{-T}{{\mathbf{H}}_{\mathcal{H}}}^{T}\mathbf{x}.
\label{hexp}
\end{equation}

At the same time, we also define the inverse tree-building process ${{\operatorname{BuildTree}}_{\mathcal{T}}}^{-1}$ as the following expression, which is taking the leaves of $\mathcal{T}$ and restoring them to $\mathbf{x}$.
\begin{equation}
\mathbf{x}={{\operatorname{BuildTree}}_{\mathcal{T}}}^{-1}\left( \mathbf{v} \right)={{\mathbf{H}}_{\mathcal{H}}}\mathbf{v}.
\label{hexp}
\end{equation}

Although most works\cite{dpht:bta,dpht:uhm,dpht:dps,dpht:dpt,dpht:pef} do not take matrix analysis as the theoretical basis for optimally consistent release, there are many advantages to applying matrix analysis. One of them is error analysis. Using matrix analysis, we can quickly calculate the overall mean square error of the ``Node Query'' after the post-processing for optimally consistent release.
\begin{theorem}
	Given the privacy budget $\varepsilon$ and the to-be-released hierarchical tree $\mathcal{T}$ containing $n$ nodes and $m$ leaves, whose height is $h$, the overall mean square error before and after post-processing $mse\left( \widetilde{\mathbf{v}} \right)$ and $mse\left( \overline{\mathbf{v}} \right)$ satisfy
	\begin{equation}
	\operatorname{mse}\left( \widetilde{\mathbf{v}} \right)=\sum\nolimits_{i=1}^{n}{\mathbb{E}\left( {{\left( {{\widetilde{v}}_{i}}-{{v}_{i}} \right)}^{2}} \right)}={2n{{h}^{2}}}/{{{\varepsilon }^{2}}}\;,
	\label{hexp}
	\end{equation}
	\begin{equation}
	\operatorname{mse}\left( \overline{\mathbf{v}} \right)=\sum\nolimits_{i=1}^{n}{\mathbb{E}\left( {{\left( {{\overline{v}}_{i}}-{{v}_{i}} \right)}^{2}} \right)}={2m{{h}^{2}}}/{{{\varepsilon }^{2}}}\;.
	\label{hexp}
	\end{equation}
	\label{thm:mse}
\end{theorem}

\begin{proof}
	%	See Appendix \ref{prov:thm:mse}.
	See Appendix A.7.
\end{proof}
According to Thm. \ref{thm:mse}, the overall mean square error depends on the number of leaves after post-processing. Generally, $n$ is much less than the number of leaves $m$, so post-processing will significantly reduce the error. As the proof shown in Appendix A.7, We applied the property of the trace of the matrix to obtain a concrete and concise demonstration, which embodies matrix analysis's great potential for solving problems.

Next, we will introduce how to apply Generation Matrix to achieve an efficient enough algorithm.

\subsection{$\text{``LO''}-$QR Decomposition on ${{\mathbf{M}}_{\mathcal{T}}}$}

To obtain an efficient release algorithm, we apply QR decomposition to analyze Formula \eqref{exp:ic}. Unfortunately, the traditional QR decomposition\cite{mx:laa} cannot meet our analysis requirements, so we propose another QR decomposition form, namely the $\text{``LO''}-$QR decomposition.

\begin{definition}[$\text{``LO''}-$QR Decomposition]
	For matrix $\mathbf{M}\in {{\mathbb{R}}^{n\times m}}\left( n\ge m \right)$, QR decomposition looks for an orthogonal matrix $\mathbf{Q}$, which converts $\mathbf{M}$ into a form composed of lower triangular matrix $\mathbf{L}$ and zero matrix. i.e.,
	\begin{equation}
	\mathbf{M}=\mathbf{Q}\left[ \begin{matrix}
	\mathbf{L}  \\
	\mathbf{O}  \\
	\end{matrix} \right].	
	\label{qr2}
	\end{equation}
	\label{def:lqrd}
\end{definition}

Correspondingly, we call the traditional QR decomposition the $\text{``UO''}-$QR decomposition, which decomposes a matrix into an upper triangular matrix. Although the two have different forms, they both achieve decomposition through a series of basic orthogonal transformations. Householder transformation is the most widely used among various orthogonal transformation techniques because of its high efficiency, easy implementation, and applicability to sparse matrices. Completing QR decomposition once requires $m$ householder transformations. To describe the QR decomposition process in more detail, we define $\left( \mathcal{S},j \right)-$Householder transformation to describe each transformation.

\begin{definition}[$\left( \mathcal{S},j \right)-$Householder Transformation]
	For matrix $\mathbf{M}$, given an ordered set $\mathcal{S}=\left\langle {{s}_{1}},{{s}_{2}},\ldots ,{ {s}_{r}} \right\rangle $ and column number $j$, $\left( \mathcal{S},j \right)-$Householder Transformation selects the rows ${{s}_{1}},{{s}_{2}},\ldots ,{{s}_{r}}$ and column $j$ of $\mathbf{M}$ as the reference for householder transformation, $\mathbf{Y}=\mathbf{QM}$. The transformation result will make $\mathbf{Y}$ satisfy ${{y}_{{{s}_{1}},j}}=\sqrt{\sum\nolimits_{s\in \mathcal{S}}{{{m}_{s,j}}^{2}}}$ and ${{y}_{{{s}_{i}},j}}=0,2\le i\le r $, where $\mathbf{Q}$ satisfies the following expression:
	\begin{equation}
	\left\{ \begin{aligned}
	& \mathbf{m}={{\mathbf{H}}_{\mathcal{S}}}\mathbf{M}{{\mathbf{e}}_{j}} \\ 
	& \boldsymbol{\omega }=\mathbf{m}-\left\| \mathbf{m} \right\|{{\mathbf{e}}_{1}} \\ 
	& \mathbf{Q}=\mathbf{I}-2{{\mathbf{H}}_{\mathcal{S}}}^{T}\frac{\boldsymbol{\omega }{{\boldsymbol{\omega }}^{T}}}{{{\boldsymbol{\omega }}^{T}}\boldsymbol{\omega }}{{\mathbf{H}}_{\mathcal{S}}} \\ 
	\end{aligned} \right..
	\label{exp:hst}
	\end{equation}
	\label{def:hst}
\end{definition}

For $\mathbf{M}$, we denote  $\left( \mathcal{S},j \right)-$Householder transformation as  $\mathbf{Y}={{\operatorname{House}}_{\mathcal{S},j}}\left( \mathbf{M} \right)$. The QR decomposition based on the Householder transformation can be expressed as
\begin{equation}
\mathbf{R}={{\operatorname{House}}_{{{\varsigma }_{m}}}}\centerdot \cdots \centerdot  {{\operatorname{House}}_{{{\varsigma }_{2}}}}\centerdot {{\operatorname{House}}_{{{\varsigma }_{1}}}}\left( \mathbf{M} \right),
\label{exp:hsQR}
\end{equation}
where $\mathbf{R}$ is the result of QR decomposition and ``$\centerdot$'' represents the operation of function composition such as ${{f}_{2}}\centerdot {{f}_{1}}\left( x \right)={{f}_{2}}\left( {{f}_{1}}\left( x \right) \right)$. For $\text{``UO''}-$QR decomposition, the parameter ${{\varsigma }_{i}}=\left( \left\langle i,\cdots ,n \right\rangle ,i \right)$; for $\text{``LO''}-$QR decomposition, the parameter ${{\varsigma }_{i}}=\left( \left\langle m-i+1,1,\cdots ,m-i,m+1,\cdots ,n \right\rangle ,m-i+1 \right)$.

Next, we apply the $\text{``LO''}-$QR decomposition on ${{\mathbf{M}}_{\mathcal{T}}}$. To understand how $\text{``LO''}-$QR decomposition affects ${{\mathbf{M}}_{\mathcal{T}}}$, we divide ${{\mathbf{M}}_{\mathcal{T}}}$ into the following forms, 
\begin{equation}
{{\mathbf{M}}_{\mathcal{T}}}\equiv \left[ \begin{matrix}
{{\mathbf{M}}_{\uparrow }}  \\
{{\mathbf{M}}_{\downarrow }}  \\
\end{matrix} \right].
\label{hexp}
\end{equation}
The upper half ${{\mathbf{M}}_{\uparrow }}\in {{\mathbb{R}}^{{{n}_{1}}\times {{n}_{1}}}}$ of ${{\mathbf{M}}_{\mathcal{T}}}$ consists of the first ${{n}_{1}}$ rows; The second half ${{\mathbf{M}}_{\downarrow }}\in {{\mathbb{R}}^{m\times {{n}_{1}}}}$ consists of the remaining $m$ rows.
According to Def. \ref{def:cmm}, we have ${{\mathbf{M}}_{\uparrow }}={{\mathbf{G} }_{{{\mathcal{T}}^{\left( 1 \right)}}}}$ and ${{\mathbf{M}}_{\downarrow }}=-{{\mathbf{ H}}_{{{\mathcal{S}}^{\left( f \right)}}}}$, 
where the $i$-th element of ordered set ${{\mathcal{S}}^{\left( f \right)}}$ satisfies $s_{i}^{\left( f \right)}=\text{ }{{f}_{i+{{n}_{1}}}}$.
Considering the effects of Householder transformation on ${{\mathbf{M}}_{\mathcal{T}}}$, we denote the matrix obtained after the $k$-th Householder transformation as $\mathbf{M}_{\mathcal{T}}^{\left( k \right)}\left( 0\le k\le {{n}_{1}} \right)$, whose upper half and second half are denoted as $\mathbf{M}_{\uparrow }^{\left( k \right)}$ and $\mathbf{M}_{\downarrow }^{\left( k \right)}$.
$\mathbf{M}_{\mathcal{T}}^{\left( 0 \right)}$ is the form before Householder transformation, satisfying  $\mathbf{M}_{\mathcal{T} }^{\left( 0 \right)}={{\mathbf{M}}_{\mathcal{T}}}$; $\mathbf{M}_{\mathcal{T}}^{\left( {{n}_{1}} \right)}$ is the result after $\text{``LO''}-$QR decomposition. Thm. \ref{thm:lqrps} demonstrates that, in the process of Householder transformation, $\mathbf{M}_{\mathcal{T}}^{\left( k \right)}$ satisfies some invariant properties.

\begin{figure}[h]
	\centering
	\includegraphics[width=6.4in,trim=0 340 0 0,clip]{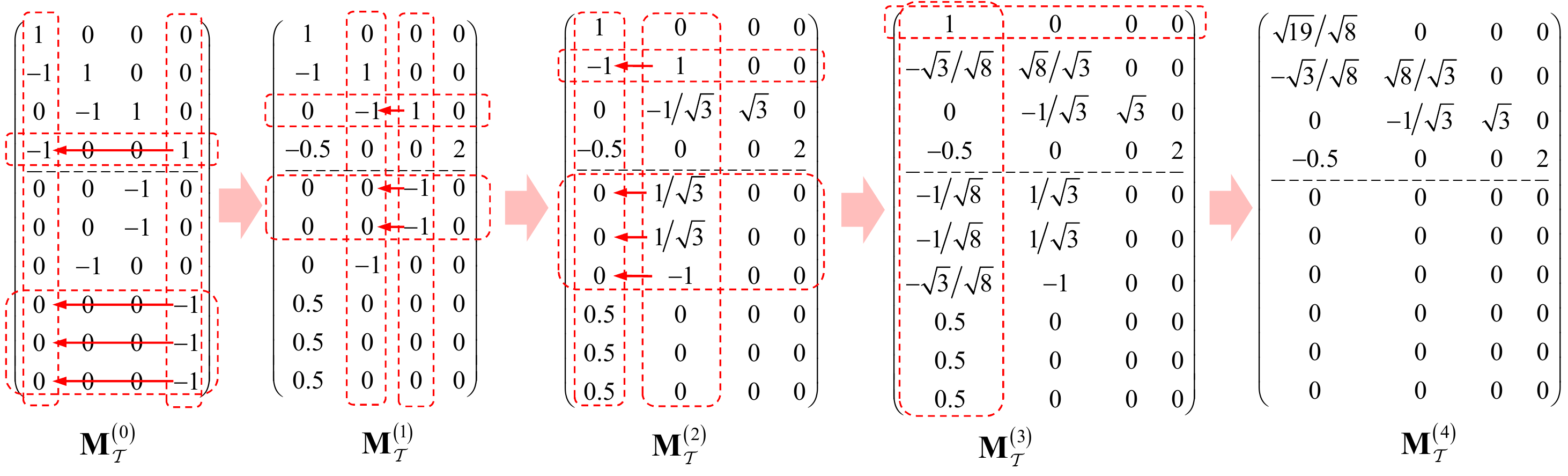}
	% where an .eps filename suffix will be assumed under latex,
	% and a .pdf suffix will be assumed for pdflatex; or what has been declared
	% via \DeclareGraphicsExtensions.
	\caption{Householder Transformations in $\text{``LO''}-$QR Decomposition on ${{\mathbf{M}}_{\mathcal{T}}}$}
	\label{fig:qrd}
\end{figure}

\begin{theorem}
	For the $\text{``LO''}-$QR decomposition on ${{\mathbf{M}}_{\mathcal{T}}}$, after $k$ $\left( 0\le k\le {{n}_{1}}-1 \right)$ times of Householder transformation, $\mathbf{M }_{\mathcal{T}}^{\left( k \right)}$ always keeps the following three properties unchanged:	
	\begin{itemize}
		\itemindent 0em
		\item[a)]$\mathbf{M}_{\uparrow }^{\left( k \right)}\sim {{\mathbf{G}}_{{{\mathcal{T}}^{\left( 1 \right)}}}}$;
		\item[b)]Each line of $\mathbf{M}_{\downarrow }^{\left( k \right)}$ has one and only one non-zero value;	
		\item[c)]The last $k$ columns of $\mathbf{M}_{\downarrow }^{\left( k \right)}$ (i.e., the columns from ${{n}_{1}}-k+1$ to ${{n}_{1}}$) are all $0$.	
	\end{itemize}
	\label{thm:lqrps}
\end{theorem}
\begin{proof}
	%	See Appendix \ref{prov:thm:lqrps}.
	See Appendix A.8.
\end{proof}

As shown in Fig. \ref{fig:qrd}, with the Householder transformation proceeds, the non-zero elements in $\mathbf{M}_{\downarrow }^{\left( k \right)}$ shift from right to left, and the position of the non-zero elements of $\mathbf{M}_{\uparrow }^{\left( k \right)}$ remains unchanged. Specifically, in the $k$-th Householder transformation, all non-zero elements in the $\left( {{n}_{1}}-k+1 \right)$-th column of $\mathbf{M}_{\downarrow }^{\left( k \right)}$ are transferred to the ${{f}_{{{n}_{1}}-k+1}} $-th column. Combining Thm. \ref{thm:lqrps}, we can infer the specific form of $\mathbf{M}_{\mathcal{T}}^{\left( {{n}_{1}}-1 \right)}$ after ${{n}_{1}}-1$ times of Householder transformation. After the last householder transformation, the result satisfies the following theorem.

\begin{theorem}
	After $\text{``LO''}-$QR decomposition, there is a $\mathbf{G}_{{{\mathcal{T}}^{\left( 1 \right)}}}^{\left( {{\mathbf{w}}_{node}},{{\mathbf{w}}_{edge}} \right)}$ and an orthogonal matrix $\mathbf{Q}$, which let  ${{\mathbf{M}}_{\mathcal{T}}}$ satisfies:
	\begin{equation}
	{{\mathbf{M}}_{\mathcal{T}}}=\mathbf{Q}\left[ \begin{matrix}
	\mathbf{G}_{{{\mathcal{T}}^{\left( 1 \right)}}}^{\left( {{\mathbf{w}}_{node}},{{\mathbf{w}}_{edge}} \right)}  \\
	\mathbf{O}  \\
	\end{matrix} \right].
	\label{hexp}
	\end{equation}
	\label{thm:gmqr}
\end{theorem}
\begin{proof}
	%	See Appendix \ref{prov:thm:gmqr}.
	See Appendix A.9.
\end{proof}

Thm. \ref{thm:gmqr} implies that the calculation ${{\left( {{\mathbf{M}}_{\mathcal{T}}}^{T}{{\mathbf{M}}_{\mathcal{T}}} \right)}^{-1}}$ about ${{\mathbf{M}}_{\mathcal{T}}}$ in expression \eqref{exp:ic} can be replaced by ${{\mathcal{T}}^{\left( 1 \right)}}$. Since it is related to $\mathcal{T}$ and ${{\mathcal{T}}^{\left( 1 \right)}}$ simultaneously, we denote it as ${{\mathbf{G}}_{{{\mathcal{T}}^{\left( 1 \right)}}\leftarrow \mathcal{T}}}$. Thm. \ref{thm:ipeqr} shows that ${{\mathbf{G}}_{{{\mathcal{T}}^{\left( 1 \right)}}\leftarrow \mathcal{T}}}$ and ${{\mathbf{M}}_{\mathcal{T}}}$ are inner-product-equivalent.

\begin{theorem}
	For arbitrary ${{\mathbf{M}}_{\mathcal{T}}}$, there is a  ${{\mathbf{G}}_{{{\mathcal{T}}^{\left( 1 \right)}}\leftarrow \mathcal{T}}}$ \textbf{inner-product-equivalent} to it, which satisfies
	\begin{equation}
	{{\mathbf{M}}_{\mathcal{T}}}^{T}{{\mathbf{M}}_{\mathcal{T}}}={{\mathbf{G}}_{{{\mathcal{T}}^{\left( 1 \right)}}\leftarrow \mathcal{T}}}^{T}{{\mathbf{G}}_{{{\mathcal{T}}^{\left( 1 \right)}}\leftarrow \mathcal{T}}},
	\label{exp:ipe}
	\end{equation}
	where ${{\mathbf{G}}_{{{\mathcal{T}}^{\left( 1 \right)}}\leftarrow \mathcal{T}}}$ is the upper half $\mathbf{M}_{\downarrow }^{\left( {{n}_{1}} \right)}$ of ${{\mathbf{M}}_{\mathcal{T}}}$ after $\text{``LO''}-$QR decomposition.
	\label{thm:ipeqr}
\end{theorem}
\begin{proof}
	%	See Appendix \ref{prov:thm:ipeqr}.
	See Appendix A.10.
\end{proof}

\subsection{Generation Matrix-based Optimally Consistent Release Algorithm}
The property of inner product equivalence provides us with a vital optimization idea for an optimal release, as shown in Cor. \ref{cor:fgic}. Using matrix analysis, we convert the expression \eqref{exp:ic} into an expression about ${{\mathbf{G}}_{{{\mathcal{T}}^{\left( 1 \right)}}\leftarrow \mathcal{T}}}$ and then use the mathematical properties of Generation Matrix to improve the efficiency of optimal release.

\begin{corollary}
	The expression $\mathbf{y}={{\left( {{\mathbf{M}}_{\mathcal{T}}}^{T}{{\mathbf{M}}_{\mathcal{T}}} \right)}^{-1}}\mathbf{x}$ can be obtained by performing an upward propagation and a downward propagation successively about ${{\mathbf{G}}_{{{\mathcal{T}}^{\left( 1 \right)}}\leftarrow \mathcal{T}}}$. That is 
	\begin{equation}
	\mathbf{y}={{\left( {{\mathbf{M}}_{\mathcal{T}}}^{T}{{\mathbf{M}}_{\mathcal{T}}} \right)}^{-1}}\mathbf{x}={{\mathbf{G}}_{{{\mathcal{T}}^{\left( 1 \right)}}\leftarrow \mathcal{T}}}^{-1}\left( {{\mathbf{G}}_{{{\mathcal{T}}^{\left( 1 \right)}}\leftarrow \mathcal{T}}}^{-T}\mathbf{x} \right).
	\label{hexp}
	\end{equation}
	\label{cor:fgic}
\end{corollary}
\begin{proof}
	%	See Appendix \ref{prov:cor:fgic}.
	See Appendix A.11.
\end{proof}
According to Cor. \ref{cor:fgic}, we get another optimal release form as follows.
\begin{equation}
\overline{\mathbf{v}}=\widetilde{\mathbf{v}}-{{\mathbf{M}}_{\mathcal{T}}}\left( {{\mathbf{G}}_{{{\mathcal{T}}^{\left( 1 \right)}}\leftarrow \mathcal{T}}}^{-1}\left( {{\mathbf{G}}_{{{\mathcal{T}}^{\left( 1 \right)}}\leftarrow \mathcal{T}}}^{-T}\left( {{\mathbf{M}}_{\mathcal{T}}}^{T}\widetilde{\mathbf{v}} \right) \right) \right).
\label{exp:fic}
\end{equation}

We illustrate with parentheses that Formula \eqref{exp:fic} is calculated from right to left, ensuring that all multiplication and solution equations are for matrices and vectors. According to the sparsity of ${{\mathbf{M}}_{\mathcal{T}}}$, the time complexity of ${{\mathbf{M}}_{\mathcal{T}}}$-related multiplication calculation in Formula \eqref{exp:fic} is $O\left( n \right)$. Besides, according to Prop. \ref{prop:up} and Prop. \ref{prop:dp}, the time complexity of solving the linear equation of Generation Matrix is also $O\left( {{n}_{1}} \right)$. Therefore, the overall time complexity of Formula \eqref{exp:fic} is $O\left( n \right)$. Formula \eqref{exp:fic} has completely summarized the core process of optimally consistent release. So long as we execute the formula directly after constructing ${{\mathbf{G}}_{{{\mathcal{T}}^{\left( 1 \right)}}\leftarrow \mathcal{T}}}$, we can efficiently obtain the optimal consistent release. The algorithm description with Generation Matrices is very concise and easy to implement.

However, only ensuring that the calculation of ${{\mathbf{G}}_{{{\mathcal{T}}^{\left( 1 \right)}}\leftarrow \mathcal{T}}}$ is also highly efficient, the whole process is efficient. So, we further proposed Thm. \ref{thm:svgm} to calculate it.

\begin{theorem}
	Let ${{w}_{i}}$$\left( 1\le i\le {{n}_{1}} \right)$ and ${{w}_{i\to {{f}_{i}}}}$$\left( 2\le i\le {{n}_{1}} \right)$ as the weights of the nodes and edges in ${{\mathbf{G}}_{{{\mathcal{T}}^{\left( 1 \right)}}\leftarrow \mathcal{T}}}$, respectively. They satisfy
	\begin{equation}
	\left\{ \begin{aligned}
	& {{w}_{i}}=\sqrt{1+{{\theta }_{i}}} \\ 
	& {{w}_{i\to {{f}_{i}}}}={{w}_{i}}^{-1} \\ 
	\end{aligned} \right.,
	\label{hexp}
	\end{equation}
	where ${{\theta }_{i}}\left( 1\le i\le {{n}_{1}} \right)$ satisfies
	\begin{equation}
	{{\theta }_{i}}=\left| {{\mathcal{C}}_{i}} \right|-\sum\limits_{j\in {{\mathcal{C}}_{i}}\wedge j\le {{n}_{1}}}{{{\left( 1+{{\theta }_{j}} \right)}^{-1}}}.
	\label{hexp}
	\end{equation}
	\label{thm:svgm}
\end{theorem}
\begin{proof}
	%	See Appendix \ref{prov:thm:svgm}.
	See Appendix A.12.
\end{proof}

Thm. \ref{thm:svgm} shows that ${{w}_{i}}$ and ${{w}_{i\to {{f}_{i}}}}$ can be directly calculated by only requiring ${{\theta }_{i}}$, and the calculation of ${{\theta }_{i}}$ only needs to traverse from ${{n}_{1}}$ to $1$ once. Since this process involves the calculation of $\left| {{\mathcal{C}}_{i}} \right|$, the overall time complexity of constructing ${{\mathbf{G}}_{{{\mathcal{T}}^{\left( 1 \right)}}\leftarrow \mathcal{T}}}$ is $O\left( n \right)$. The specific construction process is shown in Alg. \ref{alg:bGM}.

\begin{algorithm}
	\newcommand{\LASTCON}{\item[\algorithmiclastcon]}
	\footnotesize
	\caption{Construct ${{\mathbf{G}}_{{{\mathcal{T}}^{\left( 1 \right)}}\leftarrow \mathcal{T}}}$ from ${{\mathbf{M}}_{\mathcal{T}}}$}
	\renewcommand{\algorithmicrequire}{\textbf{Input:}}% 更改输入名称
	\renewcommand{\algorithmicensure}{\textbf{Output:}}% 更改输出名称
	\label{alg:bGM}
	\begin{algorithmic}[1]
		\REQUIRE Consistency Constraint Matrix ${{\mathbf{M}}_{\mathcal{T}}}$
		\ENSURE the inner-product-equivalent ${{\mathbf{G}}_{{{\mathcal{T}}^{\left( 1 \right)}}\leftarrow \mathcal{T}}}$
		\STATE Initialize $\boldsymbol{\theta }\in {{\mathbb{R}}^{{{n}_{1}}\times 1}}$, satisfying ${{\theta }_{i}}=\left | {\mathcal{C}}_{i} \right|$.
		\STATE \textbf{for} $i={{n}_{1}}\text{ to 2}$ \textbf{do} ${{\theta }_{{{f}_{i}}}}\leftarrow {{\theta }_{{{f}_{i}}}}-{{\left( 1+{{\theta }_{i}} \right)}^{-1}}$;
		\STATE \textbf{for} $i=1\text{ to }{{n}_{1}}$ \textbf{do} Let ${{w}_{i}}=\sqrt{1+{{\theta }_{i}}}$;
		\STATE \textbf{for} $i=2\text{ to }{{n}_{1}}$ \textbf{do} Let ${{w}_{i\to {{f}_{i}}}}={{w}_{i}}^{-1}$;
		\STATE Construct ${{\mathbf{G}}_{{{\mathcal{T}}^{\left( 1 \right)}}\leftarrow \mathcal{T}}}$ by ${{w}_{i}}$ and ${{w}_{i\to {{f}_{i}}}}$;
		\RETURN ${{\mathbf{G}}_{{{\mathcal{T}}^{\left( 1 \right)}}\leftarrow \mathcal{T}}}$;
	\end{algorithmic}
\end{algorithm}

In summary, we propose a Generation Matrix-based optimally consistent release algorithm (GMC) for differentially private hierarchical trees, described as Alg. \ref{alg:GIC}. Note that Step 1 to Step 3 in Alg. \ref{alg:GIC} are the normal hierarchical tree release process, and Step 4 is the call of Alg. \ref{alg:bGM}. Only step 5 is the core step, which uses formula \eqref{exp:fic} to achieve optimally consistent release. In the previous, the time complexity of formula \eqref{exp:fic} has been proved as $O\left( n \right)$. Therefore, the overall time complexity of GMC is also $O\left( n \right)$. Besides, it can be seen that GMC is a two-stage algorithm, i.e., the construction of ${{\mathbf{G}}_{{{\mathcal{T}}^{\left( 1 \right)}}\leftarrow \mathcal{T}}}$ and post-processing. If the same hierarchical tree is used for multiple releases, we only need to construct ${{\mathbf{G}}_{{{\mathcal{T}}^{\left( 1 \right)}}\leftarrow \mathcal{T}}}$ once.

\begin{algorithm}
	\newcommand{\LASTCON}{\item[\algorithmiclastcon]}
	\footnotesize
	\caption{Generation Matrix-based Optimally Consistent Release Algorithm}
	\renewcommand{\algorithmicrequire}{\textbf{Input:}}% 更改输入名称
	\renewcommand{\algorithmicensure}{\textbf{Output:}}% 更改输出名称
	\label{alg:GIC}
	\begin{algorithmic}[1]
		\REQUIRE hierarchical tree $\mathcal{T}$, dataset $\mathcal{D}$, privacy parameters $\varepsilon $
		\ENSURE the optimally consistent release $\overline{\mathbf{v}}$
		\STATE Construct a vector $\mathbf{x}$ from $\mathcal{D}$, which satisfies ${{x}_{i}}={{\phi }_{i}}\left( \mathbf{ D} \right)$.
		\STATE Build a hierarchical tree with $\mathbf{v}={{\mathbf{G}}_{\mathcal{T}}}^{-T}{{\mathbf{H}}_{\mathcal{H}}}^{T}\mathbf{x}$;
		\STATE Calculate the height of $\mathcal{T}$, then get $\widetilde{\mathbf{v}}$ by adding noise to $\mathbf{v}$, where ${{\widetilde{v}}_{i}}={{v}_{i}}+{{\xi }_{i}},{{\xi }_{i}}\sim \operatorname{Lap}\left( {h }/{\varepsilon }\; \right)$.		
		\STATE Construct ${{\mathbf{M}}_{\mathcal{T}}}$, then substitute it into Alg. \ref{alg:bGM} to obtain ${{\mathbf{G}}_{{{\mathcal{T} }^{\left( 1 \right)}}\leftarrow \mathcal{T}}}$;
		\STATE Calculate $\overline{\mathbf{v}}=\widetilde{\mathbf{v}}-{{\mathbf{M}}_{\mathcal{T}}}\left( {{\mathbf{G}}_{{{\mathcal{T}}^{\left( 1 \right)}}\leftarrow \mathcal{T}}}^{-1}\left( {{\mathbf{G}}_{{{\mathcal{T}}^{\left( 1 \right)}}\leftarrow \mathcal{T}}}^{-T}\left( {{\mathbf{M}}_{\mathcal{T}}}^{T}\widetilde{\mathbf{v}} \right) \right) \right)$;
		\RETURN $\overline{\mathbf{v}}$;
	\end{algorithmic}
\end{algorithm}

In addition, we can further optimize the algorithm. Considering that the construction of ${{\mathbf{G}}_{{{\mathcal{T}}^{\left( 1 \right)}}\leftarrow \mathcal{T}}}$ involves square root extraction, which may cause more calculation overhead, we propose an improved version of Alg. \ref{alg:GIC} to avoid any square root. The main idea is shown as follows.
\begin{equation}
\overline{\mathbf{v}}=\widetilde{\mathbf{v}}-{{\mathbf{M}}_{\mathcal{T}}}\left( \mathbf{G}{{_{{{\mathcal{T}}^{\left( 1 \right)}}}^{\left( \boldsymbol{\theta }',\mathbf{1} \right)}}^{-1}}\left( \boldsymbol{\theta }'*\left( \mathbf{G}{{_{{{\mathcal{T}}^{\left( 1 \right)}}}^{\left( \boldsymbol{\theta }',\mathbf{1} \right)}}^{-T}}\left( {{\mathbf{M}}_{\mathcal{T}}}^{T}\widetilde{\mathbf{v}} \right) \right) \right) \right).
\label{hexp}
\end{equation}
Where the vector $\boldsymbol{\theta }'={{\left[ {{\theta }_{1}},{{\theta }_{2}},\ldots ,{{\theta }_{{{n}_{1}}}} \right]}^{T}}+1$, and ``$*$'' is Hadamard product \cite{mx:mdc}. The new version without square root extraction can slightly improve the calculation efficiency, which shows that the algorithm under the matrix description has high scalability. The improvement of the algorithm only needs to make some slight modifications on Formula \eqref{exp:fic}. Furthermore, we can directly extend the existing model to solve other hierarchical tree problems, such as the hierarchical tree release with the non-uniform privacy budget.

\section{Experiment}

We conducted experiments to verify the performance of GMC our proposed for differentially private hierarchical tree release. All our experiments are run on a computer with Dual $4$ Core $3.9$ GHz AMD Ryzen CPUs, $32$GB RAM, and MATLAB's development software. To improve the reliability of our experimental results, we repeatedly run the same experimental setup $100$ times and then take the average of multiple experimental results as the final results. In addition, we denote the output of our algorithm by ${{\mathbf{v}}^{\left( out \right)}}$ or ${{\mathbf{x}}^{\left( out \right)}}$, $ v_{i}^{\left( out \right)}$, and their $i$-th outputs are denoted as $v_{i}^{\left( out \right)}$ and $ x_{i}^{\left( out \right)}$, respectively.

\subsection{Datasets and Comparison Algorithms}
Our experiments run on $3$ large datasets with more than $10$ million nodes. They are Census2010, NYCTaxi, and SynData, with the following details.

\textbf{Census2010}\cite{dt:csf}: Due to the plan that the U.S. Census Bureau announced using differential privacy\cite{dt:tuc}, we adopt the 2010 U.S. Census dataset as our first dataset, which contains demographic information of all Americans. The statistical results are divided into $8$ levels by the geographic components, i.e., ``United States - State - County - County Subdivision - Place/Remainder - Census Tract - Block Group - Block''. The dataset contains $312,471,327$ individuals, which construct a hierarchical tree containing $11,802,162$ nodes. One of its typical applications is to provide users with queries on the population of the area of interest, called ``Node Query''. For example, a user submits a query request ``What is the population of Albany County in New York, USA?''. The system will return the value at the node ``USA - New York - Albany County''. After verification, we ensured that the data before adding noise satisfies consistency.

\textbf{NYCTaxi}\cite{dt:nyc}: This data set comes from car ride records in New York City in 2013. In order to ensure the uniqueness of the data, we selected the data provided by Creative Mobile Technologies. According to the start time of the taxi, we constructed a statistical histogram of the frequency of people taking a taxi in seconds and provided a range query. The process of building a hierarchical tree is random, which ensures that our algorithm can handle any hierarchical tree structure. The data contains $86,687,775$ individuals. Since we count the ride frequency by seconds, the number of histograms (leaf nodes) totals $31,536,000$. The fan-out of nodes is random. In our experiments, we set the proportion of nodes with fan-outs of $2$, $3$, $4$ and $5$ to $\left\{ 40\%,30\%,20\%,10\% \right\}$. In this way, the hierarchical tree we build contains approximately $45$ to $50$ million nodes.

\textbf{SynData}: To test the hierarchical tree with special structures, we adopt a randomly synthesized dataset for our experiments. In synthetic data, the hierarchical tree structure is a complete binary, and the number of nodes $n$ is controlled by the tree height $h$, where $n={{2}^{h}}-1$. We generate ${{x}_{i}}$ by Poisson distribution, i.e., ${{x}_{i}}\sim \operatorname{Poi}\left( \lambda  \right)$. In our experiments, we set $\lambda =100$ and take $h=24$ as the complete dataset, containing $16,777,215$ nodes. Like NYCTaxi, we focus on the ``Range Query'' on SynData.

\begin{table}[h]
	\centering
	\caption{Description of the Algorithms}
	\label{tab:doa}
	\tabcolsep 3pt
	\begin{tabular}{*{2}{c}}
		\toprule[1.5pt]
		\multirow{1}{2.5cm}{\centering \textbf{Name}}&\multirow{1}{13cm}{\centering \textbf{Description}}\\\hline
		
		\multirow{1}{2.5cm}{\centering \textbf{Processless}}&\multirow{1}{13cm}{No processing after adding noise.}\\\hline
		
		\multirow{2}{2.5cm}{\centering \textbf{Boosting}}&\multirow{2}{13cm}{Boosting\cite{dpht:bta}, the classic post-processing of only for the complete trees. For arbitrary structure, we take the fan-out of root as the fan-out of the algorithm.}\\
		\multicolumn{1}{c}{}&\multicolumn{1}{c}{} \\\hline
		
		\multirow{2}{2.5cm}{\centering \textbf{PrivTrie}}&\multirow{2}{13cm}{The post-processing in PrivTrie\cite{dpht:pef}, which has been proven to achieve the optimally consistent release for arbitrary hierarchical tree.}\\
		\multicolumn{1}{c}{}&\multicolumn{1}{c}{} \\\hline
		
		\multirow{2}{2.5cm}{\centering \textbf{GMC}}&\multirow{2}{13cm}{Our post-processing for arbitrary hierarchical tree, which achieve the optimally consistent release based on Generation Matrix.}\\
		\multicolumn{1}{c}{}&\multicolumn{1}{c}{} \\
		\bottomrule[1.5pt]
	\end{tabular}
\end{table}
Our experiments compare $4$ different algorithm settings. Their details are shown in Tab. \ref{tab:doa}. Because our experiment is for large-scale hierarchical trees with more than $10$ million nodes, the time complexity of the selected algorithm settings is both $O\left( n \right)$.

\subsection{Verifying the Effectiveness for GMC}

In the first experiment, we verify the effectiveness of GMC with two main aspects, i.e., error and consistency. 

The error is measured by the root mean square error ($rmse$), whose calculation methods are different for ``Node Query'' and ``Range Query''. We denote the $rmse$ of ``Node Query'' (For Census2010) by ${{rmse}_{NQ}}$, and its formula is
\begin{equation}
{{rmse}_{NQ}}=\sqrt{\frac{1}{n}\sum\nolimits_{i=1}^{n}{{{\left( v_{i}^{\left( out \right)}-{{v}_{i}} \right)}^{2}}}}.
\label{exp:rmse1}
\end{equation}
And the $rmse$ of ``Range Query'' (For NYCTaxi and SynData) is denoted as ${{rmse}_{RQ}}$, calculated by
\begin{equation}
{{rmse}_{RQ}}=\sqrt{\frac{1}{q}\sum\nolimits_{i=1}^{q}{{{\left( \sum\nolimits_{j={{a}_{i}}}^{{{b}_{i}}}{x_{i}^{\left( out \right)}}-\sum\nolimits_{j={{a}_{i}}}^{{{b}_{i}}}{{{x}_{j}}} \right)}^{2}}}}.
\label{exp:rmse1}
\end{equation}
Where $q$ is the number of range queries selected randomly and without repetition, and the range of the $i$-th query is recorded as $\left[ {{a}_{i}},{{b}_ {i}} \right]$. The main reason for random sampling is that all range queries are up to $C_{m}^{2}$ so that we cannot test all range queries. In our experiment, we take $q={{10}^{5}}$.

The consistency of the outputs is measured by consistency bias. Let $\mathbf{\Delta }={{\mathbf{M}}^{T}}{{\mathbf{v}}^{\left( out \right)}}$, and ${{\Delta }_{i}}$ as the $i$-th element of $\mathbf{\Delta }$, then $bias$ satisfies
\begin{equation}
bias=\sqrt{{\sum\nolimits_{i=1}^{{{n}_{1}}}{{{\Delta }_{i}}^{2}}}/{{{n}_{1}}}\;}.
\label{exp:rmse1}
\end{equation}

Besides, we adopt the complete datasets in the experiment and take $\varepsilon =1$.
\begin{table}[h]
	\centering
	\caption{Experimental Results on Multiple Algorithm Settings and Large-scale Datasets}
	\label{tab:exp1}
	\begin{tabular}{|c|r|r|r|r|r|r|}
		\hline
		\multicolumn{1}{|l|}{\multirow{2}{*}{}} & \multicolumn{2}{c|}{Census2010}                      & \multicolumn{2}{c|}{NYCTaxi}                        & \multicolumn{2}{c|}{SynData}                         \\ \cline{2-7} 
		\multicolumn{1}{|l|}{}                  & \multicolumn{1}{c|}{${{rmse}}$} & \multicolumn{1}{c|}{${{bias}}$} & \multicolumn{1}{c|}{${{rmse}}$} & \multicolumn{1}{c|}{${{bias}}$} & \multicolumn{1}{c|}{${{rmse}}$} & \multicolumn{1}{c|}{${{bias}}$} \\ \hline
		\textbf{Processless} &11.32	&49.61	&138.42	&48.38	&166.37	&61.24
		\\ \hline
		\textbf{Boosting}&11.30	&159.01	&120.75	&22.42	&\textbf{74.94}	&\textbf{0.00}\\ \hline
		\textbf{PrivTrie}&\textbf{11.00}	&\textbf{0.00}	&\textbf{68.69}	&\textbf{0.00}	&\textbf{74.94}	&\textbf{0.00}
		\\ \hline
		\textbf{GMC}&\textbf{11.00}	&\textbf{0.00}	&\textbf{68.69}	&\textbf{0.00}	&\textbf{74.94}	&\textbf{0.00}\\
		\hline
	\end{tabular}
\end{table}

Since the optimally consistent release problem of the differentially private hierarchical tree is a convex optimization problem, its solution is unique. If the outputs of GMC also satisfy optimally consistent, they should be the same as PrivTrie. The results in Tab. \ref{tab:exp1} confirm this point. They show that GMC is effective and correct. However, it does not mean that PrivTrie and GMC are entirely equivalent. Their implementations are entirely different, so we need further to analyze the algorithms' performance in the subsequent experiments.

In addition, the results also show the major drawback of Boosting. i.e., it can only guarantee the consistency of complete trees. If the fan-outs of nodes are different, Boosting cannot guarantee that the results are consistent. Especially for Census2010, where the fan-out of nodes is highly different, Boosting results in more significant consistency bias after post-processing.

\subsection{Performance Testing}

In this section, we focus on the algorithms' performance and test the running time of the algorithms above from small-scale data to large-scale data. To construct hierarchical trees with different scales, we adopted the following different methods according to the characteristics of each dataset. For Census2010, we obtain a smaller hierarchical tree by $k-$order subtree. For example, in its $4-$order subtree, the leaves represent the ``County Subdivision'' level. For NYC taxi, we divide the data of 2013 into $12$ months. The data of the first $k$ months as the $k$-th subset. For SynData, we control the data scale by directly setting the tree height. In the experiment, we used the six different tree heights $\left\{ 8,12,16,20,24 \right\}$ to generate hierarchical trees with different scales. Since Processless does not perform any processing for consistency, we omit it in the experiment.

\begin{figure*}[!h]
	\subfigure[Census2010]{
		\begin{minipage}[t]{5.3cm}
			\centering
			\includegraphics[width=2.3in,trim=20 20 50 60,clip]{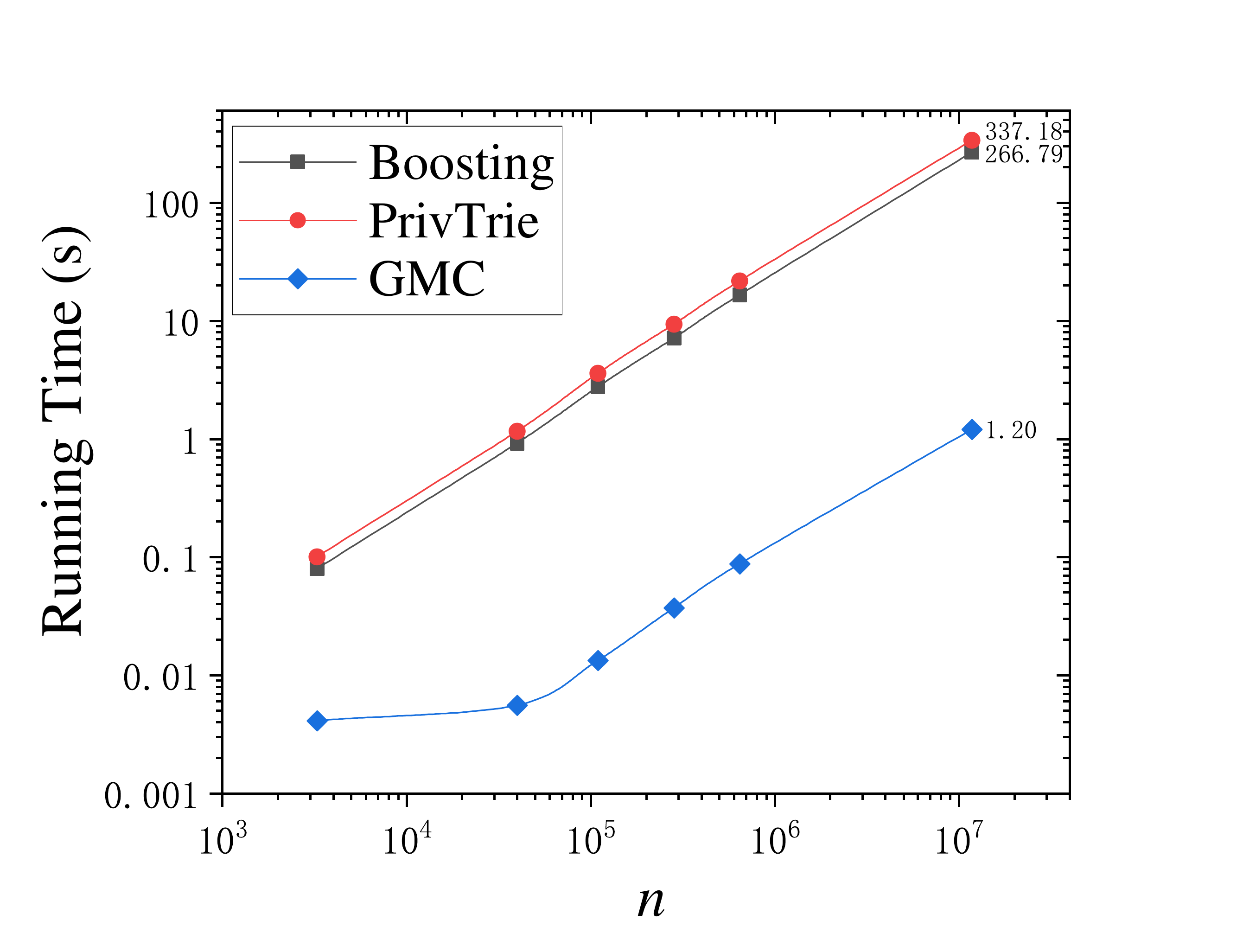}
		\end{minipage}%
	}%
	\subfigure[NYCTaxi]{
		\begin{minipage}[t]{5.4cm}
			\centering
			\includegraphics[width=2.3in,trim=20 20 50 60,clip]{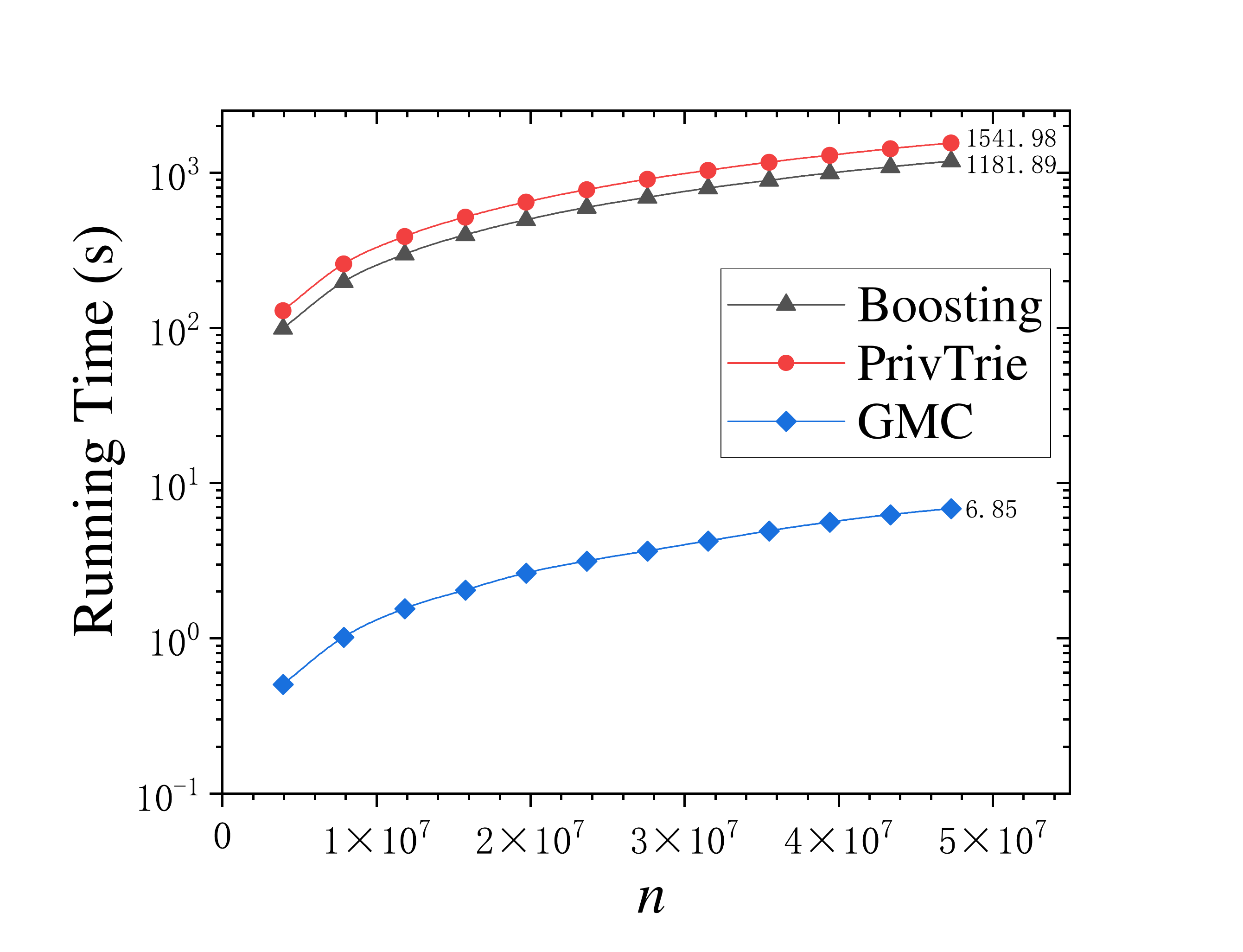}
		\end{minipage}%
	}%
	\subfigure[SynData]{
		\begin{minipage}[t]{5.4cm}
			\centering
			\includegraphics[width=2.3in,trim=20 20 50 60,clip]{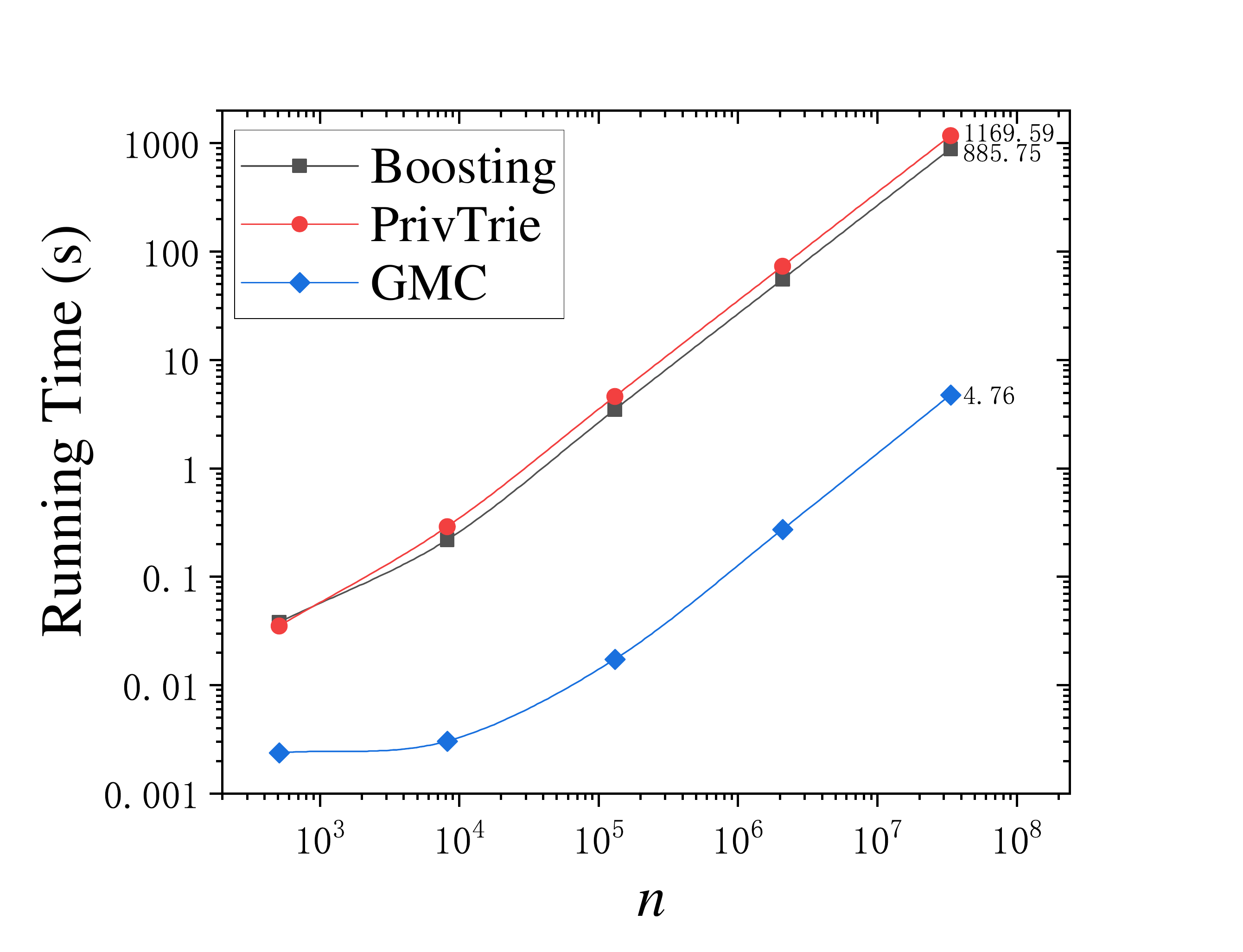}
		\end{minipage}%
	}%
	\caption{The Comparison of Running Time for Various Different Data Scales}
	\label{fig:expVs}
\end{figure*}
In Fig. \ref{fig:expVs}, the experimental results show that all three algorithms can complete the post-processing within some time, but their running times are quite different. Boosting and PrivTrie need more than 200 seconds to process the hierarchical trees with tens of millions of nodes, while GMC only needs about $2$ seconds to process the same data (See Fig. \ref{fig:expVs}a). Their performance gap is up to $100$ times. It shows that even if Boosting and PrivTrie have reached the lowest time complexity with $O\left( n \right)$, they still have much space for performance improvement. In addition to the relatively inefficient recursive algorithms, another important reason is that GMC uses standard matrix operations to complete the post-processing after establishing the matrix model. These standard matrix operations introduce many optimization techniques in the underlying design, which can make full use of computer resources and significantly improve computing efficiency. Although the results above do not deny that Boosting and PrivTrie can handle large-scale hierarchical trees, applying them to some scenarios, such as real-time data updates or more complex models, may cause considerable challenges in computational efficiency.

\begin{figure}[!h]
	\centering
	\includegraphics[width=2.5in,trim=0 10 0 60,clip]{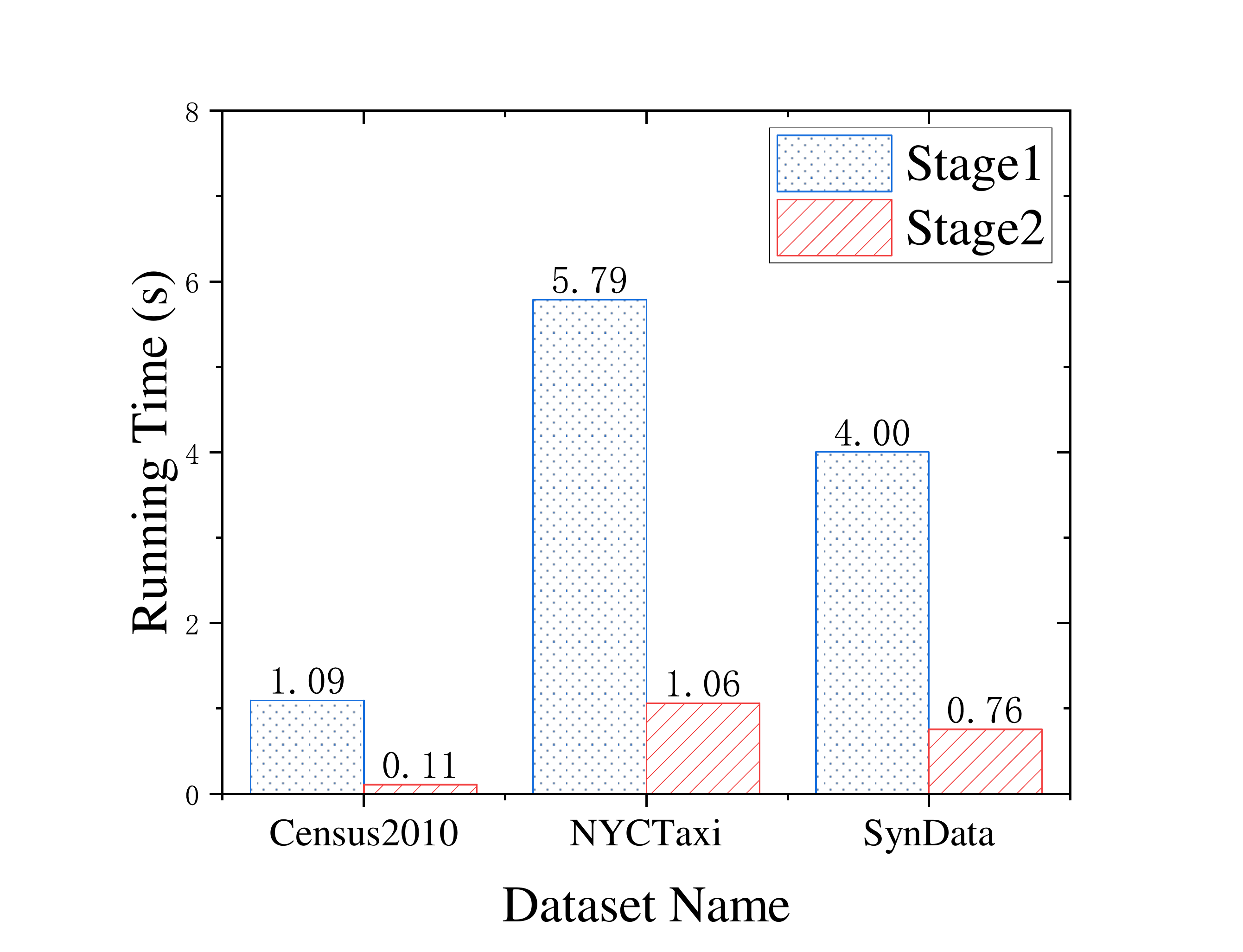}
	% where an .eps filename suffix will be assumed under latex,
	% and a .pdf suffix will be assumed for pdflatex; or what has been declared
	% via \DeclareGraphicsExtensions.
	\caption{Running Time for Two Stages of GMC with Complete Datasets}
	\label{fig:exp2sg}
\end{figure}

Finally, we test the running time for the two stages of GMC, i.e., the Generation Matrix construction (Stage 1) and the post-processing (Stage 2). In Fig. \ref{fig:exp2sg}, we can see that the time overhead in Stage 1 is much more significant than in Stage 2. The reason is the Generation Matrix construction includes the sorting by Descending Order of Height, counting the number of children, and some high-cost operations such as division, square root (in the process of calculating ${{\theta }_{i}}$). It is worth noting that the Generation Matrix construction is independent of the data to be released and does not involve individual privacy. When the same hierarchical tree structure is reused in multiple releases, we only need to perform the construction process once, further reducing the overall running time of the optimally consistent releases.

\section{Conclusions and Future Work}

In the previous, we successfully defined Generation Matrix and demonstrated many of its critical mathematical properties. Using Generation Matrix, we can implement various hierarchical tree operations without accessing the local structure, which provides crucial theoretical support for the Matrixing Research Method of hierarchical trees. The application on the differentially private hierarchical tree release reflects the Generation Matrix's practicability. The proposed GMC based on Generation Matrix provides a concise algorithm for optimally consistent release. Our experiments show that GMC has achieved a significant performance improvement of up to $100$ times compared with the state-of-the-art schemes.

The scientific problem that we solve in this paper is not very complicated, but it is classic and is suitable as an example to show the practicability of the Generation Matrix. However, the hierarchical tree problems that Generation Matrix can solve are far more than the application. We can use it to explore more complex hierarchical tree release problems and even the problems that have not been solved so far.  For example, the issue of non-negative consistent release of hierarchical trees currently does not exist any closed-form solution with time complexity of $O\left( n \right)$, which makes it very difficult to solve the optimal release satisfying both consistency and non-negativity for large-scale data. Nonetheless, Generation Matrix provides us with a critical analysis tool to challenge the problem.

%\iffalse
\appendix
\section{Partial Proofs}

\subsection{Proof of Property \ref{prop:inv}}\label{prov:prop:inv}
\begin{proof}
	According to the basic properties of the lower triangular matrix\cite{mx:tar}, the inverse ${{\mathbf{G}}_{\mathcal{T}}}^{-1}$ of the lower triangular matrix ${{\mathbf{G}}_{\mathcal{T}}}$ is also a lower triangular matrix.
	
	Considering ${{\mathbf{G}}_{\mathcal{T}}}^{-1}$'s $j$-th column vector $\mathbf{g}_{j}^{\left( -1 \right)}={{\mathbf{G}}_{\mathcal{T}}}^{-1}{{\mathbf{e}}_{j}}$, whose $i$-th element is denoted as $g_{ij}^{\left( -1 \right)}$, we have $g_{ij}^{\left( -1 \right)}=0,i<j$. Since only the $j$-th element of ${{\mathbf{e}}_{j}}$ is $1$, but the rest are all $0$, $g_{ij}^{\left( -1 \right)}$, for $i\ge j$, satisfies
	\begin{equation}
	g_{ij}^{\left( -1 \right)}=\left\{ \begin{array}{*{35}{l}}
	1 & ,i=j  \\
	g_{{{f}_{j}},j}^{\left( -1 \right)} & ,i>j  \\
	\end{array} \right..
	\label{rec:ginv}
	\end{equation}
	
	Therefore, $g_{jj}^{\left( -1 \right)}=1$. Next, we adopt the contradiction method to prove. Suppose $j$ is the ancestor of $i$, but $g_{ij}^{\left( -1 \right)}\ne 1$.
	
	Let $t_{0}^{\left( i \right)},t_{1}^{\left( i \right)},\ldots ,t_{p}^{\left( i \right)}$ denote the node $i$ and its $p$ ancestors respectively.
	
	Since $j$ is an ancestor of $i$, there is a $q$ such that $j=it_{q}^{\left( i \right)}$. Therefore, we have $ g_{jj}^{\left( -1 \right)}=g_{t_{q}^{\left( i \right)}j}^{\left( -1 \right)}=\cdots =g_{t_{1}^{\left( i \right)}j}^{\left( -1 \right)}=g_{ij}^{\left( -1 \right)}\ne 1$ according to the recurrence formula \eqref{rec:ginv}.
	
	The conclusion contradicts $g_{jj}^{\left( -1 \right)}=1$. Therefore, when $j$ is an ancestor of $i$, $g_{ij}^{\left( -1 \right)}=1$.
	
	If $j$ is not an ancestor of $i$, obviously $j$ is not the root, i.e., $j>1$. According to the formula \eqref{rec:ginv}, we have $g_{1,j}^{\left( -1 \right)}=g_{t_{p}^{\left( i \right)}j}^{\left(- 1 \right)}=\cdots =g_{t_{1}^{\left( i \right)}j}^{\left( -1 \right)}=g_{ij}^{\left( -1 \right)}$. Since $g_{1,j}^{\left( -1 \right)},j>1$ has been proved to be $0$, thus $g_{ij}^{\left( -1 \right)}=0$.
\end{proof}

\subsection{Proof of Property \ref{prop:can}}\label{prov:prop:can}
\begin{proof}
	Considering the $i$-th row of ${{\mathbf{G}}_{\mathcal{T}}}^{-1}$, we denote ${{U}_{i}}=\left\{ k\left| g_{ik}^{\left( -1 \right)}\ne 0 \right. \right\}$ as the set formed by the column subscripts of the non-zero elements in the $i$-th row. According to $\mathbf{M}={{\left( {{\mathbf{G}}_{\mathcal{T}}}^{T}{{\mathbf{G}}_{\mathcal{T}}} \right)}^{-1}}={{\mathbf{G}}_{\mathcal{T}}}^{-1}{{\mathbf{G}}_{\mathcal{T}}}^{-T}$, we have ${{m}_{ij}}$ satisfied
	
	\begin{equation}
	{{m}_{ij}}= \sum\nolimits_{k=1}^{n}{g_{ik}^{\left( -1 \right)}g_{jk}^{\left( -1 \right)}}=\sum\nolimits_{k\in {{U}_{i}}\cap {{U}_{j}}}{g_{ik}^{\left( -1 \right)}g_{jk}^{\left( -1 \right)}} = \left| {{U}_{i}}\cap {{U}_{j}} \right|  
	\label{hexp}
	\end{equation}
	
	According to Prop. \ref{prop:inv}, ${{U}_{i}}=\left\{ k\left| k\text{ is an ancestor of }i \right. \right\}$. Therefore, ${{U}_{i}}\cap {{U}_{j}}$ is the common ancestor of $i$ and $j$ and ${{m}_{ij}}$ records the number of their common ancestors.
	
	Specifically, if $i=j$, ${{m}_{ii}}=\left| {{U}_{i}} \right|$, i.e., the number of ancestors of $i$ plus $1$ (itself). Let the depth of the root be 1, ${{m}_{ii}}$ is the depth of $i$.
\end{proof}

\subsection{Proof of Property \ref{prop:dd}}\label{prov:prop:dd}
\begin{proof}
	Let $\boldsymbol{\alpha }={{\left[ {{\alpha }_{1}},{{\alpha }_{2}},\ldots ,{{\alpha }_{n}} \right]}^{T}}$, $\boldsymbol{\beta }={{\left[ {{\beta }_{1}},{{\beta }_{2}},\ldots ,{{\beta }_{n}} \right]}^{T}}$. According to Equation \eqref{eq:ddEq}, we have ${{w}_{i}}={{\alpha }_{i}}{{\beta }_{i}}$, ${{w}_{i\to {{f}_{i}}}}={{\alpha }_{{{f}_{i}}}}{{\beta }_{i}}$. That is, $\ln {{w}_{i}}=\ln {{\alpha }_{i}}+\ln {{\beta }_{i}}$, $\ln {{w}_{i\to {{f}_{i}}}}={{\alpha }_{{{f}_{i}}}}+{{\beta }_{i}}$. Therefore, we have the following matrix equation holds.
	\begin{equation}
	\left[ \begin{matrix}
	\mathbf{I} & \mathbf{I}  \\
	\mathbf{I}-\mathbf{G} & \mathbf{I}  \\
	\end{matrix} \right]\left[ \begin{matrix}
	\ln \boldsymbol{\alpha }  \\
	\ln \boldsymbol{\beta }  \\
	\end{matrix} \right]=\left[ \begin{matrix}
	\ln {{\mathbf{w}}_{node}}  \\
	\ln {{\mathbf{w}}_{edge}}  \\
	\end{matrix} \right].
	\label{hexp}
	\end{equation}
	
	According to the property of Block Matrix Inversion, we have
	\begin{equation}
	{{\left[ \begin{matrix}
			\mathbf{I} & \mathbf{I}  \\
			\mathbf{I}-\mathbf{G} & \mathbf{I}  \\
			\end{matrix} \right]}^{-1}}=\left[ \begin{matrix}
	{{\mathbf{G}}^{-1}} & -{{\mathbf{G}}^{-1}}  \\
	\mathbf{I}-{{\mathbf{G}}^{-1}} & {{\mathbf{G}}^{-1}}  \\
	\end{matrix} \right].
	\label{hexp}
	\end{equation}
	And then,
	\begin{equation}
	\left[ \begin{matrix}
	\ln \boldsymbol{\alpha }  \\
	\ln \boldsymbol{\beta }  \\
	\end{matrix} \right]= \left[ \begin{matrix}
	{{\mathbf{G}}_{\mathcal{T}}}^{-1}\left( \ln {{\mathbf{w}}_{node}}-\ln {{\mathbf{w}}_{edge}} \right)  \\
	\ln {{\mathbf{w}}_{node}}+{{\mathbf{G}}_{\mathcal{T}}}^{-1}\left( \ln {{\mathbf{w}}_{edge}}-\ln {{\mathbf{w}}_{node}} \right)  \\
	\end{matrix} \right] = \left[ \begin{matrix}
	\ln \boldsymbol{\alpha }  \\
	\ln {{\mathbf{w}}_{node}}-\ln \boldsymbol{\alpha }  \\
	\end{matrix} \right].
	\label{eq:lnab}
	\end{equation}
	
	Performing the exponential operations $\exp \left( * \right)$ on both sides of the equation \eqref{eq:lnab} at the same time, we have formula \eqref{eq:svab} holds.
\end{proof}

\subsection{Proof of Theorem \ref{thm:clm}}\label{prov:thm:clm}
\begin{proof}
	Let ${{l}_{ij}}$ denote the element in row $i$ and column $j$ of ${{\mathbf{L}}_{\mathcal{T}}}$. According to the definition of Laplacian Matrix\cite{mxp:twm}, we have the diagonal elements ${{l}_{ii}}$ representing the number of adjacent nodes of $i$; and if $j$ is the adjacent node of $i$, i.e., $j={{f}_{i}}$ or $i={{f}_{j}}$, we have ${{l}_{ij}}=-1$; otherwise, ${{l} _{ij}}=0$.
	
	According to the expression \eqref{exp:gclm}, we have
	\begin{equation}
	{{l}_{ij}}=\left\{ \begin{array}{*{35}{l}}
	\left( \sum\nolimits_{k\in {{U}_{1}}}{{{g}_{k1}}{{g}_{k1}}} \right)-1 & ,i=j=1  \\
	\sum\nolimits_{k\in {{U}_{i}}\cap {{U}_{j}}}{{{g}_{ki}}{{g}_{kj}}} & ,\text{otherwise}  \\
	\end{array} \right.
	\end{equation}
	where ${{U}_{j}}=\left\{ k\left| {{g}_{kj}}\ne 0 \right. \right\}=\left\{ j \right\}\cup {{\mathcal{C}}_{j}}$ denotes the set formed by the row subscripts of non-zero elements in the $j$-th column of ${{\mathbf{G}}_{\mathcal{T}}}$.
	
	Next, we discuss the following situations:
	
	For $i=j=1$, under the definition of Laplacian Matrix, ${{l}_{11}}$ should be equal to the number of children of the root, i.e., ${{l}_{11 }}=\left| {{\mathcal{C}}_{1}} \right|$. 	 According to the expression \eqref{exp:gclm}, there is $k\in \left\{ 1 \right\}\cup {{\mathcal{C}}_{1}}$ and ${{l}_{11}}=\left| \left\{ 1 \right\}\cup {{\mathcal{C}}_{1}} \right|-1=\left| {{\mathcal{C}}_{1}} \right|$. It conforms to the definition of Laplacian Matrix.
	
	For $i=j>1$, under the definition of Laplacian Matrix, ${{l}_{ii}}$ should be equal to the number of $i$'s children and parent, i.e., ${ {l}_{ii}}=\left| {{\mathcal{C}}_{i}}\cup {{f}_{i}} \right|=\left| {{\mathcal{C} }_{i}} \right|+1$. According to the expression \eqref{exp:gclm}, ${{l}_{ii}}=\sum\nolimits_{k\in {{U}_{i}}}{{{g}_{ki}}{{g}_ {ki}}}=\left| \left\{ i \right\}\cup {{\mathcal{C}}_{i}} \right|=\left| {{\mathcal{C}}_{ i}} \right|+1$. It conforms to the definition of Laplacian Matrix.
	
	For $i\ne j$ but $i\in {{\mathcal{C}}_{j}}$, under the definition of Laplacian Matrix, ${{l}_{ij}}=-1$. According to the expression \eqref{exp:gclm}, $k\in {{U}_{i}}\cap {{U}_{j}}\text{=}\left( \left\{ i \right\}\cup { {\mathcal{C}}_{i}} \right)\cap \left( \left\{ j \right\}\cup {{\mathcal{C}}_{j}} \right)=\left \{ i \right\}\cap {{\mathcal{C}}_{j}}=\left\{ i \right\}$, we have ${{l}_{ij}}={{ g}_{ii}}{{g}_{ij}}=-1$. Again, It conforms to the definition of Laplacian Matrix. Similarly, if $i\ne j$ and $j\in {{\mathcal{C}}_{i}}$, ${{l}_{ij}}=-1$.
	Finally, when $i\ne j$ and $j$ and $i$ do not contain a parent-child relationship, ${{U}_{i}}\cap {{U}_{j}}=\varnothing $ , ${{L}_{ij}}=0$, which also conforms to the definition of Laplacian Matrix.
	
	Therefore, in any case, the expression\eqref{exp:gclm} always conforms to the definition of the Laplacian Matrix.
\end{proof}

\subsection{Proof of Theorem \ref{thm:cdm}}\label{prov:thm:cdm}
\begin{proof}
	Let ${{d}_{ij}}$ denote the element in row $i$ and column $j$ of ${{\mathbf{D}}_{\mathcal{T}}}$. 
	
	\begin{figure}[h]
		\centering
		\includegraphics[width=1.5in,trim=0 480 0 0,clip]{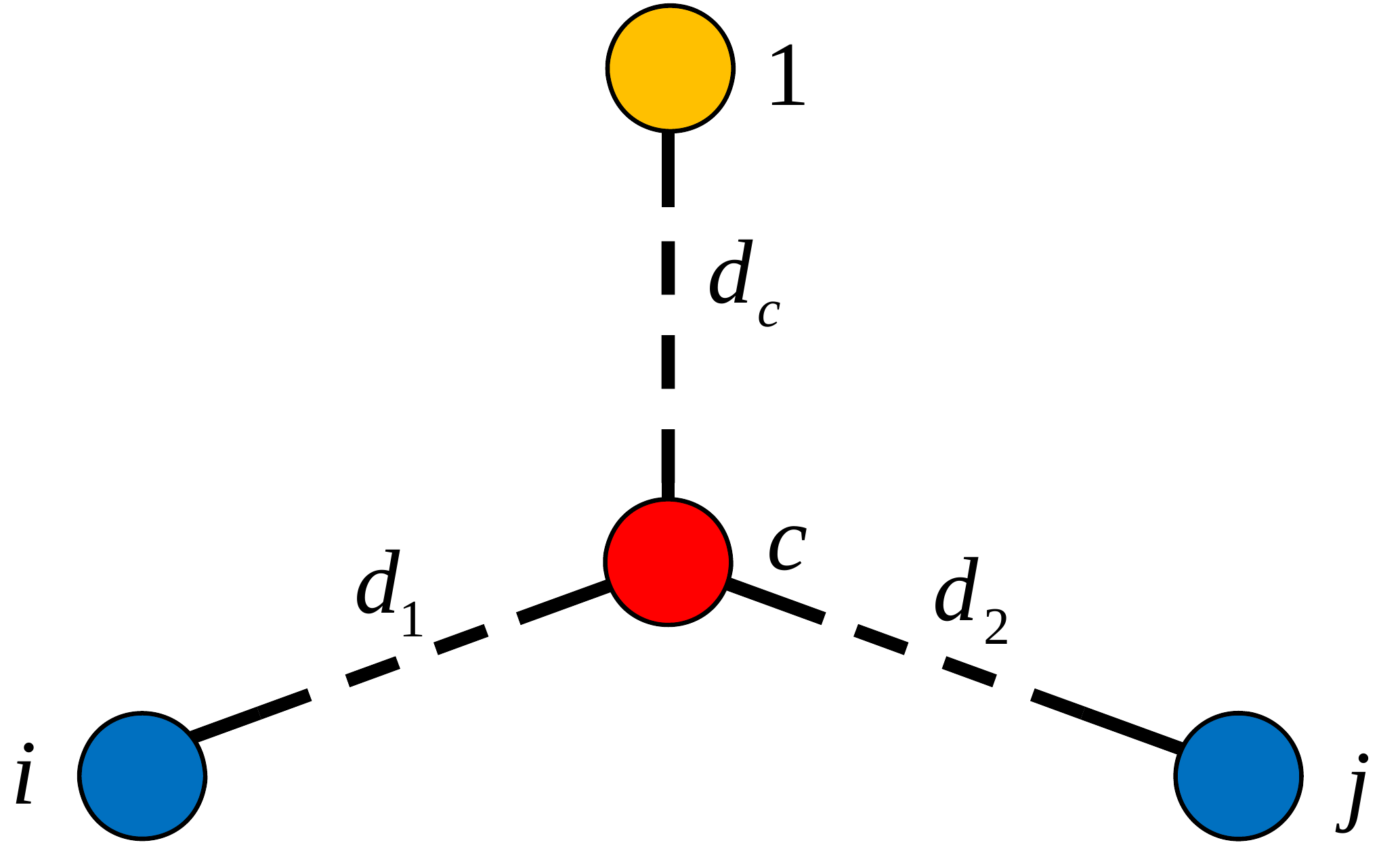}
		% where an .eps filename suffix will be assumed under latex,
		% and a .pdf suffix will be assumed for pdflatex; or what has been declared
		% via \DeclareGraphicsExtensions.
		\caption{The Distance of Node $i$, $j$ and Their Nearest Common Ancestor}
		\label{fig:pr1}
	\end{figure}
	According to the definition of distance matrix\cite{mxp:sdm}, ${{d}_{ij}}$ is the distance from node $i$ to $j$. As shown in Fig. \ref{fig:pr1}, $c$ (red node) is the nearest common ancestor of nodes $i$ and $j$. ${{d}_{c}}$ denotes the depth of node $c$ (i.e., the distance from node C to the root (orange node) plus 1); ${{d}_{1}}$ denotes the distance from $c$ to $i$; ${{d}_{2}}$ denotes the distance from $c$ to $j$. Obviously, the distance ${{d}_{ij}}={{d}_{1}}+{{d}_{2}}$.
	
	Considering expression \eqref{exp:gcdm}, let ${{\mathbf{M}}_{1}}={{\mathbf{G}}_{\mathcal{T}}}^{-1}\mathbf{I}{{\mathbf{I}}^{T}}$, ${{\mathbf{M}}_{2}}=\mathbf{I}{{\mathbf{I}}^{T}}{{\mathbf{G}}_{\mathcal{T}}}^{-T}$, ${{\mathbf{M}}_{2}}=\mathbf{I}{{\mathbf{I}}^{T}}{{\mathbf{G}}_{\mathcal{T}}}^{-T}$ and $m_{ij}^{\left( k \right)}$ denotes the element in the row $i$ and column $j$ of ${{\mathbf{M}}_{k}}$. According to Prop. \ref{prop:depn}, $m_{ij}^{\left( 1 \right)}=m_{ji}^{\left( 2 \right)}$ is the depth of the node $i$, i.e., $m_{ij}^{\left( 3 \right)}={{d}_{c}}$. According to Prop. \ref{prop:can}, $m_{ij}^{\left( 3 \right)}$ is the number of common ancestors of $i$ and $j$ (the depth of $c$), i.e., $m_{ij}^{\left( 3 \right)}={{d}_{c}}$.
	
	Substituting them into expression \eqref{exp:gcdm}, we have
	\begin{equation}
	{{d}_{ij}}= m_{ij}^{\left( 1 \right)}+m_{ij}^{\left( 2 \right)}-2m_{ij}^{\left( 3 \right)} = \left( {{d}_{1}}+{{d}_{c}} \right)+\left( {{d}_{2}}+{{d}_{c}} \right)-2{{d}_{c}} = {{d}_{1}}+{{d}_{2}}
	\label{exp:dij}
	\end{equation}
	The expression  \eqref{exp:dij} is consistent with the definition, so ${{\mathbf{D}}_{\mathcal{T}}}$ can be calculated by expression \eqref{exp:gcdm}.
\end{proof}
\subsection{Proof of Theorem \ref{thm:ccm}}\label{prov:thm:ccm}
\begin{proof}
	Let ${{c}_{ij}}$ denote the element in row $i$ and column $j$ of ${{\mathbf{C}}_{\mathcal{T}}}$. 
	
	According to the definition of Ancestral Matrix\cite{mxp:tam}, ${{c}_{ij}}$ represents the distance from the nearest common ancestor of leaves $i$ and $j$ to the root.
	
	Note, Prop. \ref{prop:can} has been proved that the element of ${{\left( {{\mathbf{G}}_{\mathcal{T}}}^{T}{{\mathbf{G}}_{\mathcal {T}}} \right)}^{-1}}$ is the number of the common nodes, i.e., the distance from the nearest common ancestor to the root node plus $1$.
	
	Therefore, we can get ${{\mathbf{C}}_{\mathcal{T}}}$ by taking the sub-matrix corresponding to the leaves in ${{\left( {{\mathbf{G}}_{\mathcal{T}}}^{T}{{\mathbf{G}}_{\mathcal{T}}} \right)}^{-1}}$ and then subtracting $1$.
\end{proof}

\subsection{Proof of Theorem \ref{thm:mse}}\label{prov:thm:mse}
\begin{proof}
	According to formula \eqref{dp:addNoi}, the noise we add to each element in $\widetilde{\mathbf{v}}$ is i.i.d, and satisfies $\operatorname{Lap}\left( {h}/{ \varepsilon }\; \right)$. Since the variance of $\operatorname{Lap}\left( b \right)$ is $2{{b}^{2}}$, the covariance matrix  of $\widetilde{\mathbf{v}}$ is $D\left( \widetilde{\mathbf{v}} \right)=2\left( {{{h}^{2}}}/{{{\varepsilon }^{2}}}\; \right)\mathbf{I}$.
	
	According to the mean square error analysis of differential privacy, we have
	\begin{equation}
	\operatorname{mse}\left( \widetilde{\mathbf{v}} \right)= \operatorname{trace}\left( D\left( \widetilde{\mathbf{v}} \right) \right)=2\left( {{{h}^{2}}}/{{{\varepsilon }^{2}}}\; \right)\operatorname{trace}\left( \mathbf{I} \right) = {2n{{h}^{2}}}/{{{\varepsilon }^{2}}}\;
	\end{equation}
	
	In addition, according to formula \eqref{exp:ic},
	\begin{equation}
	D\left( \overline{\mathbf{v}} \right)= D\left( \left( \mathbf{I}-{{\mathbf{M}}_{\mathcal{T}}}{{\left( {{\mathbf{M}}_{\mathcal{T}}}^{T}{{\mathbf{M}}_{\mathcal{T}}} \right)}^{-1}}{{\mathbf{M}}_{\mathcal{T}}}^{T} \right)\widetilde{\mathbf{v}} \right) = 2\left( {{{h}^{2}}}/{{{\varepsilon }^{2}}}\; \right)\left( \mathbf{I}-{{\mathbf{M}}_{\mathcal{T}}}{{\left( {{\mathbf{M}}_{\mathcal{T}}}^{T}{{\mathbf{M}}_{\mathcal{T}}} \right)}^{-1}}{{\mathbf{M}}_{\mathcal{T}}}^{T} \right)  
	\end{equation}
	Then,
	\begin{equation}
	\begin{aligned}
	 \operatorname{mse}\left( \overline{\mathbf{v}} \right)=& \operatorname{trace}\left( D\left( \overline{\mathbf{v}} \right) \right) = 2\left( {{{h}^{2}}}/{{{\varepsilon }^{2}}}\; \right)\operatorname{trace}\left( \mathbf{I}-{{\mathbf{M}}_{\mathcal{T}}}{{\left( {{\mathbf{M}}_{\mathcal{T}}}^{T}{{\mathbf{M}}_{\mathcal{T}}} \right)}^{-1}}{{\mathbf{M}}_{\mathcal{T}}}^{T} \right) \\ 
	=& 2\left( {{{h}^{2}}}/{{{\varepsilon }^{2}}}\; \right)\left( n-\operatorname{trace}\left( {{\mathbf{M}}_{\mathcal{T}}}^{T}{{\mathbf{M}}_{\mathcal{T}}}{{\left( {{\mathbf{M}}_{\mathcal{T}}}^{T}{{\mathbf{M}}_{\mathcal{T}}} \right)}^{-1}} \right) \right) \\ 
	\end{aligned}.
	\label{exp:msev1}
	\end{equation}
	Since ${{\mathbf{M}}_{\mathcal{T}}}\in {{\mathbb{R}}^{n\times {{n}_{1}}}}$, we have
	\begin{equation}
	\operatorname{trace}\left( {{\mathbf{M}}_{\mathcal{T}}}^{T}{{\mathbf{M}}_{\mathcal{T}}}{{\left( {{\mathbf{M}}_{\mathcal{T}}}^{T}{{\mathbf{M}}_{\mathcal{T}}} \right)}^{-1}} \right)=\operatorname{trace}\left( \mathbf{I} \right)={{n}_{1}}.
	\end{equation}
	Substituting it into \eqref{exp:msev1}, we have
	\begin{equation}
	\operatorname{mse}\left( \overline{\mathbf{v}} \right)=2\left( {{{h}^{2}}}/{{{\varepsilon }^{2}}}\; \right)\left( n-{{n}_{1}} \right)={2m{{h}^{2}}}/{{{\varepsilon }^{2}}}\;
	\end{equation}
\end{proof}

\subsection{Proof of Theorem \ref{thm:lqrps}}\label{prov:thm:lqrps}
\begin{proof}
	By Def. \ref{def:lqrd}, it is obvious that the process of Householder transformation satisfies property c). Because the process of $\text{``LO''}-$QR decomposition is the process of setting $0$ column by column from the last column to the first column.
	
	According to Def. \ref{def:cmm}, $\mathbf{M}_{\mathcal{T}}^{\left( 0 \right)}$ satisfies both properties a),b) and c).
	
	Assume that M satisfies both properties a), b) and c). By the formula \eqref{exp:hsQR}, the expression of the $k$-th Householder transformation is as follows:
	\begin{equation}
	\mathbf{M}_{\mathcal{T}}^{\left( k \right)}={{\operatorname{House}}_{{{\mathcal{S}}_{k}},{{t}_{k}}}}\left( \mathbf{M}_{\mathcal{T}}^{\left( k-1 \right)} \right),
	\label{hexp}
	\end{equation}
	where ${{t}_{k}}={{n}_{1}}-k+1$. Considering that $\mathbf{M}_{\mathcal{T}}^{\left( k-1 \right)}$ satisfies the property a) and is a lower triangular matrix, the values of ${{\mathcal{S}}_{k}}$ can be simplified, taking ${{\mathcal{S}}_{k}}=\left\langle {{t}_{k}},{{n}_{1}}+1,\ldots ,n \right\rangle$.
	
	According to Def. \ref{def:hst}
	, during the transformation process, ${{\mathbf{m}}_{k}}$ and ${{\boldsymbol{\omega }}_{k}}$ satisfy
	\begin{equation}
	\left\{ \begin{aligned}
	& {{\mathbf{m}}_{k}}={{\left[ m_{{{t}_{k}},{{t}_{k}}}^{\left( k-1 \right)},m_{{{n}_{1}}+1,{{t}_{k}}}^{\left( k-1 \right)},m_{{{n}_{1}}+2,{{t}_{k}}}^{\left( k-1 \right)},\ldots ,m_{n,{{t}_{k}}}^{\left( k-1 \right)} \right]}^{T}} \\ 
	& {{\boldsymbol{\omega }}_{k}}={{\left[ m_{{{t}_{k}},{{t}_{k}}}^{\left( k-1 \right)}-\rho ,m_{{{n}_{1}}+1,{{t}_{k}}}^{\left( k-1 \right)},m_{{{n}_{1}}+2,{{t}_{k}}}^{\left( k-1 \right)},\ldots ,m_{n,{{t}_{k}}}^{\left( k-1 \right)} \right]}^{T}} \\ 
	\end{aligned} \right.,
	\label{hexp}
	\end{equation}
	where ${{\rho }_{k}}=\left\| {{\mathbf{m}}_{k}} \right\|=\sqrt{\sum\limits_{s\in {{\mathcal{S}}_{k}}}{m{{_{s,{{t}_{k}}}^{\left( k \right)}}^{2}}}}$. By calculating ${{\boldsymbol{\omega }}_{k}}^{T}{{\mathbf{m}}_{k}}$, we have
	\begin{equation}
	{{\boldsymbol{\omega }}_{k}}^{T}{{\mathbf{m}}_{k}}=\left( m_{{{t}_{k}},{{t}_{k}}}^{\left( k-1 \right)}-{{\rho }_{k}} \right)m_{{{t}_{k}},j}^{\left( k-1 \right)}+\sum\limits_{s={{n}_{1}}+1}^{n}{m_{s,{{t}_{k}}}^{\left( k-1 \right)}m_{s,j}^{\left( k-1 \right)}}.
	\label{hexp}
	\end{equation}
	
	Due to $\mathbf{M}_{\mathcal{T}}^{\left( k-1 \right)}$ satisfies the property b), for the same row $s>{{n}_{1}}$, at most only one of $ m_{s,{{t}_{k}}}^{\left( k-1 \right)}$ and $ m_{s,j}^{\left( k-1 \right)}$ in $\mathbf{M}_{\mathcal{T}}^{\left( k-1 \right)}$ is non-zero. Therefore,
	\begin{equation}
	{{\boldsymbol{\omega }}_{k}}^{T}{{\mathbf{m}}_{k}}=\left( m_{{{t}_{k}},{{t}_{k}}}^{\left( k-1 \right)}-{{\rho }_{k}} \right)m_{{{t}_{k}},j}^{\left( k-1 \right)}.
	\label{hexp}
	\end{equation}
	
	Simplifying according to formula \eqref{exp:hst}, we can get $\mathbf{M}_{\mathcal{T}}^{\left( k \right)}$ after $k$-th Householder transformation, whose element $m_{i,j}^{\left( k \right)}$ in in row $i$ and column $j$ satisfies
	\begin{equation}
	m_{i,j}^{\left( k \right)}=\left\{ \begin{array}{*{35}{l}}
	{{\rho }_{k}} & ,i={{t}_{k}},j={{t}_{k}}  \\
	0 & ,{{n}_{1}}+1\le i\le n,j={{t}_{k}}  \\
	{m_{{{t}_{k}},j}^{\left( k-1 \right)}m_{{{t}_{k}},{{t}_{k}}}^{\left( k-1 \right)}}/{{{\rho }_{k}}} & ,i={{t}_{k}},j\ne {{t}_{k}}  \\
	m_{i,j}^{\left( k-1 \right)}+{m_{i,{{t}_{k}}}^{\left( k-1 \right)}m_{{{t}_{k}},j}^{\left( k-1 \right)}}/{{{\rho }_{k}}} & ,{{n}_{1}}+1\le i\le n,j\ne {{t}_{k}}  \\
	m_{i,j}^{\left( k-1 \right)} & ,\text{otherwise}  \\
	\end{array} \right.
	\label{exp:mRs}
	\end{equation}
	
	According to the recursive expression \eqref{exp:mRs}, consider the properties of $\mathbf{M}_{\mathcal{T}}^{\left( k \right)}$. First, consider the property a). Since $\mathbf{M}_{\uparrow }^{\left( k-1 \right)}\sim {{\mathbf{G}}_{{{\mathcal{T}}^{\left( 1 \right)}}}}$, and only the row ${{t}_{k}}={{n}_{1}}-k+1$ in the first ${{n}_{1}}$ rows of $\mathbf{M}_{\mathcal{T}}^{\left( k-1 \right)}$ is affected, according to the sparsity of GM, it can be known that the ${{t}_{k}}$-th row of $\mathbf{M}_{\mathcal{T}}^{\left( k-1 \right)}$ satisfies $ m_{{{t}_{k}},j}^{\left( k-1 \right)}\ne 0$ if and only if $ j\in \left\{ {{t}_{k}},{{f}_{{{t}_{k}}}} \right\}$.
	
	Since ${m_{{{t}_{k}},{{t}_{k}}}^{\left( k-1 \right)}}/{{{\rho }_{k}}}\;\ne 0$, according to $ m_{{{t}_{k}},j}^{\left( k \right)}={m_{{{t}_{k}},j}^{\left( k-1 \right)}m_{{{t}_{k}},{{t}_{k}}}^{\left( k-1 \right)}}/{{{\rho }_{k}}}\;$ for $ j\ne {{t}_{k}}$, we have $ m_{{{t}_{k}},j}^{\left( k \right)}\ne 0\Leftrightarrow m_{{{t}_{k}},j}^{\left( k-1 \right)}\ne 0$.
	
	Therefore, there is no mutual conversion between non-zero elements and zero elements in the only affected ${{t}_{k}}$-th row. $\mathbf{M}_{\mathcal{T}}^{\left( k \right)}$ satisfies property a).
	
	Next consider property b). For the rows ${{n}_{1}}+1\sim n$, according to the property c), all the values of the ${{t}_{k}}$-th column of $\mathbf{M}_{\downarrow }^{\left( k \right)}$ are $0$. Consider the $j$-th column $\left( j\ne {{t}_{k}} \right)$ of $\mathbf{M}_{\downarrow }^{\left( k \right)}$, we have
	\begin{equation}
	m_{i,j}^{\left( k \right)}=m_{i,j}^{\left( k-1 \right)}+{m_{i,{{t}_{k}}}^{\left( k-1 \right)}m_{{{t}_{k}},j}^{\left( k-1 \right)}}/{{{\rho }_{k}}}\;.
	\label{hexp}
	\end{equation}
	
	If $j\ne {{f}_{{{t}_{k}}}}$, there is $m_{{{t}_{k}},j}^{\left( k-1 \right)}=0$, i.e., $m_{i,j}^{\left( k \right)}=m_{i,j}^{\left( k-1 \right)}$; If $j={{f}_{{{t}_{k}}}}$, there is $m_{{{t}_{k}},j}^{\left( k-1 \right)}\ne 0$, i.e., $m_{i,j}^{\left( k \right)}\ne 0\Leftrightarrow m_{i,j}^{\left( k-1 \right)}\ne 0\vee m_{i,{{t}_{k}}}^{\left( k-1 \right)}\ne 0$.
	
	Therefore, the $k$-th Householder transformation is equivalent to transferring non-zero elements from the ${{t}_{k}}$-th column of $\mathbf{M}_{\downarrow }^{\left( k-1 \right)}$ to the ${{f}_{{{t}_{k}}}}$-th column of $\mathbf{M}_{\downarrow }^{\left( k-1 \right)}$. The process keeps the property b) holds.
	
	In summary, $\mathbf{M}_{\mathcal{T}}^{\left( k \right)}$ also satisfies the properties a), b) and c). Since $\mathbf{M}_{\mathcal{T}}^{\left( 0 \right)}$ satisfies the properties a), b) and c), all $\mathbf{M}_{\mathcal{T}}^{\left( k \right)} $$\left( 0\le k\le {{n}_{1}}-1 \right) $ satisfy the properties a), b) and c).
\end{proof}

\subsection{Proof of Theorem \ref{thm:gmqr}}\label{prov:thm:gmqr}
\begin{proof}
	By Thm. \ref{thm:lqrps}, after ${{n}_{1}}-1$ Householder transformations, $\mathbf{M}_{\mathcal{T}}^{\left( {{n}_{1}}-1 \right)}$ satisfies both the properties a), b), and c).
	
	Consider $\mathbf{M}_{\mathcal{T}}^{\left( {{n}_{1}} \right)}={{\operatorname{House}}_{{{\mathcal{S}}_{{{n}_{1}}}},1}}\left( \mathbf{M}_{\mathcal{T}}^{\left( {{n}_{1}}-1 \right)} \right)$ for the ${{n}_{1}}$-th Householder transformation. Due to $\mathbf{M}_{\uparrow }^{\left( {{n}_{1}}-1 \right)}\sim {{\mathbf{G}}_{{{\mathcal{T}}^{\left( 1 \right)}}}}$ and only $m_{1,1}^{\left( {{n}_{1}}-1 \right)}$ in the first row of $\mathbf{M}_{\uparrow }^{\left( {{n}_{1}}-1 \right)}$ is non-zero, which is the only element of $\mathbf{M}_{\uparrow }^{\left( {{n}_{1}}-1 \right)}$ affected by the last transformation. The transformed $m_{1,1}^{\left( {{n}_{1}} \right)}$ satisfies
	\begin{equation}
	m_{1,1}^{\left( {{n}_{1}} \right)}=\sqrt{m{{_{1,1}^{\left( {{n}_{1}}-1 \right)}}^{2}}+\sum\limits_{s={{n}_{1}}+1}^{n}{m{{_{s,1}^{\left( {{n}_{1}}-1 \right)}}^{2}}}}.
	\label{hexp}
	\end{equation}
	
	Therefore, there is still $\mathbf{M}_{\uparrow }^{\left( {{n}_{1}} \right)}\sim {{\mathbf{G}}_{{{\mathcal{T}}^{\left( 1 \right)}}}}$ after the last transformation, i.e., there is a $\mathbf{G}_{{{\mathcal{T}}^{\left( 1 \right)}}}^{\left( {{\mathbf{w}}_{node}},{{\mathbf{w}}_{edge}} \right)}$ satisfying $\mathbf{G}_{{{\mathcal{T}}^{\left( 1 \right)}}}^{\left( {{\mathbf{w}}_{node}},{{\mathbf{w}}_{edge}} \right)}=\mathbf{M}_{\uparrow }^{\left( {{n}_{1}} \right)}$. Besides, according to Def. \ref{def:hst} and property c), it can be known that $\mathbf{M}_{\downarrow }^{\left( {{n}_{1}} \right)}=\mathbf{O}$ after the last transformation.
\end{proof}

\subsection{Proof of Theorem \ref{thm:ipeqr}}\label{prov:thm:ipeqr}

\begin{proof}
	Let ${{\mathbf{G}}_{{{\mathcal{T}}^{\left( 1 \right)}}\leftarrow \mathcal{T}}}$ be the upper half of ${{\mathbf{M}}_{\mathcal{T}}}$ after the $\text{``LO''}-$QR decomposition. According to Thm. \ref{thm:gmqr}, we have
	\begin{equation}
	{{\mathbf{M}}_{\mathcal{T}}}=\mathbf{Q}\left[ \begin{matrix}
	{{\mathbf{G}}_{{{\mathcal{T}}^{\left( 1 \right)}}\leftarrow \mathcal{T}}}  \\
	\mathbf{O} \\
	\end{matrix} \right].
	\label{hexp}
	\end{equation}
	
	Substituting it into the expression ${{\mathbf{M}}_{\mathcal{T}}}^{T}{{\mathbf{M}}_{\mathcal{T}}}$, we have
	\begin{equation}
	{{\mathbf{M}}_{\mathcal{T}}}^{T}{{\mathbf{M}}_{\mathcal{T}}}= \left[ \begin{matrix}
	{{\mathbf{G}}_{{{\mathcal{T}}^{\left( 1 \right)}}\leftarrow \mathcal{T}}}^{T} & \mathbf{O}  \\
	\end{matrix} \right]{{\mathbf{Q}}^{T}}\mathbf{Q}\left[ \begin{matrix}
	{{\mathbf{G}}_{{{\mathcal{T}}^{\left( 1 \right)}}\leftarrow \mathcal{T}}}  \\
	\mathbf{O}  \\
	\end{matrix} \right] = \left[ \begin{matrix}
	{{\mathbf{G}}_{{{\mathcal{T}}^{\left( 1 \right)}}\leftarrow \mathcal{T}}}^{T} & \mathbf{O}  \\
	\end{matrix} \right]\left[ \begin{matrix}
	{{\mathbf{G}}_{{{\mathcal{T}}^{\left( 1 \right)}}\leftarrow \mathcal{T}}}  \\
	\mathbf{O}  \\
	\end{matrix} \right] = {{\mathbf{G}}_{{{\mathcal{T}}^{\left( 1 \right)}}\leftarrow \mathcal{T}}}^{T}{{\mathbf{G}}_{{{\mathcal{T}}^{\left( 1 \right)}}\leftarrow \mathcal{T}}}.  
	\label{hexp}
	\end{equation}
\end{proof}

\subsection{Proof of Corollary \ref{cor:fgic}}\label{prov:cor:fgic}

\begin{proof}
	Since ${{\mathbf{G}}_{{{\mathcal{T}}^{\left( 1 \right)}}\leftarrow \mathcal{T}}}$ is reversible, according to Thm. \ref{thm:ipeqr}, substituting formula \eqref{exp:ipe} into $\mathbf{y}={{\left( {{\mathbf{M}}_{\mathcal{T}}}^{T}{{\mathbf{M}}_{\mathcal{T}}} \right)}^{-1}}\mathbf{x}$, we have
	\begin{equation}
	\mathbf{y}= {{\left( {{\mathbf{M}}_{\mathcal{T}}}^{T}{{\mathbf{M}}_{\mathcal{T}}} \right)}^{-1}}\mathbf{x} = {{\mathbf{G}}_{{{\mathcal{T}}^{\left( 1 \right)}}\leftarrow \mathcal{T}}}^{-1}\left( {{\mathbf{G}}_{{{\mathcal{T}}^{\left( 1 \right)}}\leftarrow \mathcal{T}}}^{-T}\mathbf{x} \right). 
	\label{hexp}
	\end{equation}
\end{proof}

\subsection{Proof of Theorem \ref{thm:svgm}}\label{prov:thm:svgm}

\begin{proof}
	Define a sequence ${{\theta }_{i}}\left( 1\le i\le {{n}_{1}} \right)$ satisfies
	\begin{equation}
	{{\theta }_{i}}=\sum\limits_{s={{n}_{1}}+1}^{n}{m{{_{s,i}^{\left( {{t}_{i}} \right)}}^{2}}},
	\label{exp:theta}
	\end{equation}
	where ${{t}_{i}}={{n}_{1}}-i$. The $\left( {{t}_{i}}+1 \right)$-th Householder transformation is the Householder transformation on the $i$-th column. And let $j$ denote the child of $i$, i.e., $\left( j\in {{\mathcal{C}}_{i}} \right)$.
	
	By recursive expression \eqref{exp:mRs}, in the first ${{t}_{i}}$ Householder transformations, only the Householder transformation on the $j$-th column (i.e., the $\left( {{t}_{j}}+1 \right)$-th transformation) will cause the value of ${{m}_{s,i}}$$\left( {{n}_{1}}+1\le s\le n \right)$ to change. Therefore, $m_{s,i}^{\left( {{t}_{i}} \right)}$ satisfies
	\begin{equation}
	m_{s,i}^{\left( {{t}_{i}} \right)}=m_{s,i}^{\left( 0 \right)}+\sum\limits_{j\in {{\mathcal{C}}_{i}}\wedge j\le {{n}_{1}}}{\left( {m_{s,j}^{\left( {{t}_{j}} \right)}m_{j,i}^{\left( {{t}_{j}} \right)}}/{m_{j,j}^{\left( {{t}_{j}}+1 \right)}}\; \right)}.
	\label{hexp}
	\end{equation}
	
	For a non-leaf node $i$, since the values of $m_{j,i}^{\left( {{t}_{j}} \right)}$ and $ m_{j,j}^{\left( {{t}_{j}} \right)}$ haven't changed in the first ${{t}_{j}}$ Householder transformations, there are $ m_{j,j}^{\left( {{t}_{j}} \right)}=1$ and $ m_{j,i}^{\left( {{t}_{j}} \right)}=-1$. Therefore, $ m_{j,j}^{\left( {{t}_{j}}+1 \right)}$ satisfies 
	\begin{equation}
	m_{j,j}^{\left( {{t}_{j}}+1 \right)}=\sqrt{m{{_{j,j}^{\left( {{t}_{j}} \right)}}^{2}}+\sum\limits_{s={{n}_{1}}+1}^{n}{m{{_{s,j}^{\left( {{t}_{j}} \right)}}^{2}}}}=\sqrt{1+{{\theta }_{j}}}.
	\label{ana:mjj}
	\end{equation}
	
	Therefore, for ${{n}_{1}}+1\le s\le n$, we have
	\begin{equation}
	m_{s,i}^{\left( {{t}_{i}} \right)}=m_{s,i}^{\left( 0 \right)}-\sum\limits_{j\in {{\mathcal{C}}_{i}}\wedge j\le {{n}_{1}}}{{m_{s,j}^{\left( {{t}_{j}} \right)}}/{\sqrt{1+{{\theta }_{j}}}}\;}.
	\label{hexp}
	\end{equation}
	
	According to formula \eqref{exp:theta}, ${{\theta }_{i}}$ satisfies
	\begin{equation}
	{{\theta }_{i}}= \sum\limits_{s={{n}_{1}}+1}^{n}{m{{_{s,i}^{\left( {{t}_{i}} \right)}}^{2}}} = \sum\limits_{s={{n}_{1}}+1}^{n}{{{\left( m_{s,i}^{\left( 0 \right)}-\sum\limits_{j\in {{\mathcal{C}}_{i}}\wedge j\le {{n}_{1}}}{{m_{s,j}^{\left( {{t}_{j}} \right)}}/{\sqrt{1+{{\theta }_{j}}}}\;} \right)}^{2}}}. 
	\label{hexp}
	\end{equation}
	
	According to Thm. \ref{thm:lqrps}, the process of $\text{``LO''}-$QR decomposition always satisfies the property b), so at most only one item in $m_{s,i}^{\left( 0 \right)}-\sum\limits_{j\in {{\mathcal{C}}_{i}}\wedge j\le {{n}_{1}}}{{m_{s,j}^{\left( {{t}_{j}} \right)}}/{\sqrt{1+{{\theta }_{j}}}}\;}$ is non-zero, and we have
	\begin{equation} {{\theta }_{i}}= {{\left( m_{s,i}^{\left( 0 \right)}-\sum\limits_{j\in {{\mathcal{C}}_{i}}\wedge j\le {{n}_{1}}}{{m_{s,j}^{\left( {{t}_{j}} \right)}}/{\sqrt{1+{{\theta }_{j}}}}\;} \right)}^{2}} = m{{_{s,i}^{\left( 0 \right)}}^{2}}+\sum\limits_{j\in {{\mathcal{C}}_{i}}\wedge j\le {{n}_{1}}}{{m{{_{s,j}^{\left( {{t}_{j}} \right)}}^{2}}}/{\left( 1+{{\theta }_{j}} \right)}\;}. 
	\label{hexp}
	\end{equation}
	
	Therefore, 
	\begin{equation}
	\begin{aligned}
	{{\theta }_{i}}=& \sum\limits_{s={{n}_{1}}+1}^{n}{{{\left( m_{s,i}^{\left( 0 \right)}-\sum\limits_{j\in {{\mathcal{C}}_{i}}\wedge j\le {{n}_{1}}}{{m_{s,j}^{\left( {{t}_{j}} \right)}}/{\sqrt{1+{{\theta }_{j}}}}\;} \right)}^{2}}} = \sum\limits_{s={{n}_{1}}+1}^{n}{m{{_{s,i}^{\left( 0 \right)}}^{2}}}+\sum\limits_{j\in {{\mathcal{C}}_{i}}\wedge j\le {{n}_{1}}}{\left( {\sum\limits_{s={{n}_{1}}+1}^{n}{m{{_{s,j}^{\left( {{t}_{j}} \right)}}^{2}}}}/{\left( 1+{{\theta }_{j}} \right)}\; \right)} \\ 
	=& \sum\limits_{s={{n}_{1}}+1}^{n}{m{{_{s,i}^{\left( 0 \right)}}^{2}}}+\sum\limits_{j\in {{\mathcal{C}}_{i}}\wedge j\le {{n}_{1}}}{{{{\theta }_{j}}}/{\left( 1+{{\theta }_{j}} \right)}\;} = \sum\limits_{s={{n}_{1}}+1}^{n}{m{{_{s,i}^{\left( 0 \right)}}^{2}}}+\sum\limits_{j\in {{\mathcal{C}}_{i}}\wedge j\le {{n}_{1}}}{\left( 1-{{\left( 1+{{\theta }_{j}} \right)}^{-1}} \right)}. \\ 
	\end{aligned}
	\label{hexp}
	\end{equation}
	
	According to the definition of ${{\mathbf{M}}_{\downarrow }}$, for ${{n}_{1}}+1\le s\le n$, if $ s\in {{\mathcal{C}}_{i}}$, then $ m_{s,i}^{\left( 0 \right)}=-1$, otherwise $ m_{s,i}^{\left( 0 \right)}=0$. Finally, ${{\theta }_{i}}$ is reduced to
	\begin{equation}
	{{\theta }_{i}}= \sum\limits_{j\in {{\mathcal{C}}_{i}}\wedge j>{{n}_{1}}}{1}+\sum\limits_{j\in {{\mathcal{C}}_{i}}\wedge j\le {{n}_{1}}}{1}-\sum\limits_{j\in {{\mathcal{C}}_{i}}\wedge j\le {{n}_{1}}}{{1}/{\left( 1+{{\theta }_{j}} \right)}\;} = \sum\limits_{j\in {{\mathcal{C}}_{i}}}{1}-\sum\limits_{j\in {{\mathcal{C}}_{i}}\wedge j\le {{n}_{1}}}{{1}/{\left( 1+{{\theta }_{j}} \right)}\;} = \left| {{\mathcal{C}}_{i}} \right|-\sum\limits_{j\in {{\mathcal{C}}_{i}}\wedge j\le {{n}_{1}}}{{1}/{\left( 1+{{\theta }_{j}} \right)}\;}. 
	\label{hexp}
	\end{equation}
	According to formula \eqref{ana:mjj}, we have
	\begin{equation}
	{{w}_{i}}=m_{i,i}^{\left( {{t}_{i}}+1 \right)}=\sqrt{1+{{\theta }_{i}}}.
	\label{hexp}
	\end{equation}
	Then, according to formula \eqref{exp:mRs}, we have
	\begin{equation}
	{{w}_{i\to {{f}_{i}}}}= -m_{{{f}_{i}},i}^{\left( {{t}_{i}}+1 \right)}=-{m_{i,{{f}_{i}}}^{\left( {{t}_{i}} \right)}m_{i,i}^{\left( {{t}_{i}} \right)}}/{m_{i,i}^{\left( {{t}_{i}}+1 \right)}}\; = -{m_{i,{{f}_{i}}}^{\left( 0 \right)}m_{i,i}^{\left( 0 \right)}}/{m_{i,i}^{\left( {{t}_{i}}+1 \right)}}\;=-{\left( -1\times 1 \right)}/{{{w}_{i}}}\; = {{w}_{i}}^{-1}. 
	\label{hexp}
	\end{equation}
\end{proof}
%\fi


%% References


\bibliographystyle{elsarticle-num}
\bibliography{refs}

\end{document}